\documentclass[reqno,dvips]{amsart}
\usepackage{graphicx}
\newcommand{\bee}{\begin{eqnarray}}
\newcommand{\eee}{\end{eqnarray}}
\newcommand{\be}{\begin{eqnarray*}}
\newcommand{\ee}{\end{eqnarray*}}
\newcommand{\R}{\mathbb {R}}
\newcommand{\I}{\mathbb {I}}
\newcommand{\C}{\mathbb {C}}
\newcommand{\N}{\mathbb {N}}
\newcommand{\E}{\mathcal E}
\newcommand{\Ss}{\mathcal{S}}
\newcommand{\Oo}{\mathcal{O}}
\newcommand{\U}{\mathcal U}
\newcommand{\Hh}{\mathcal H}
\newcommand{\K}{\mathcal K}
\newcommand{\Rc}{\mathcal R}

\newcommand{\rr}{\mathbf {r}}
\newcommand{\RR}{\mathbf {R}}
\newcommand{\LL}{\mathbf {L}}

\newcommand{\x}{\hat {\mathbf {x}}}
\newcommand{\y}{\hat {\mathbf {y}}}
\newcommand{\z}{\hat {\mathbf {z}}}

\newcommand{\D}{{\rm D}}
\def\re{\mathop{\rm Re}\nolimits}
\def\im{\mathop{\rm Im}\nolimits}
\newtheorem{theorem}{Theorem}[section]
\newtheorem{lemma}[theorem]{Lemma}
\newtheorem{proposition}[theorem]{Proposition}
\newtheorem{definition}[theorem]{Definition}
\newtheorem{remark}[theorem]{Remark}
\newtheorem{corollary}[theorem]{Corollary}

\begin{document}

\title {The Stark effect on $H_2^+$-like molecules}

\author {V. Grecchi}

\address {Vincenzo Grecchi, Dipartimento di Matematica \\ Universit\'a degli Studi di Bologna \\ Piazza di Porta San Donato 5, Bologna 40126, Italy}

\email {grecchi@dm.unibo.it}

\author {H. Kova\v{r}\'{\i}k}

\address {Hynek Kova\v{r}\'{\i}k, Dipartimento di Matematica \\ Politecnico di Torino \\Corso Duca degli Abruzzi 24, Torino
10129, Italy}

\email {hynek.kovarik@polito.it}

\author {A. Martinez}

\address {Andr\'e Martinez, Dipartimento di Matematica \\ Universit\'a degli Studi di Bologna \\ Piazza di Porta San Donato 5, Bologna 40126, Italy}

\email {martinez@dm.unibo.it}

\author {A. Sacchetti}

\address {Andrea Sacchetti, Dipartimento di Matematica Pura ed Applicata\\ Universit\'a
degli studi di Modena e Reggio Emilia\\ Via Campi 213/B, Modena 41100,
Italy}

\email {andrea.sacchetti@unimore.it}

\author {V. Sordoni}

\address {Vania Sordoni, Dipartimento di Matematica \\ Universit\'a degli Studi di Bologna \\ Piazza di Porta San Donato 5, Bologna 40126, Italy}

\email {sordoni@dm.unibo.it}

\date {\today}

\thanks {H.K. was supported by the German Research Foundation (DFG) under Grant KO 3636/1-1.}

\begin {abstract}

We consider the vibrational energy levels of the first two electronic states of the molecule ion $H_2^+$. \ The Born-Oppenheimer method applied to the case of the Stark effect on a $H_2^+$-like molecule gives existence of sharp resonances localized in the same interval of energy of the vibrational levels.

\end{abstract}


\maketitle
\hfill {\textbf{\textit{In Memory of Pierre Duclos}}
\section {Introduction}

We consider the Stark effect on a diatomic ionized molecule of the type of $H_2^+$ due to an external field.

Up to now, this problem has been treated by many authors in a heuristic way in the coaxial case where the external field and the nuclear axis of the molecule have same direction \cite{CMM,Hiskes,MPS}. \ This attitude is justified since, in such a case, the effect of the external field on the diatomic molecule is the strongest one.

Our problem consists of a non isolated three-body system, where one of the masses is much smaller than the other ones, under the effect of a  coaxial and locally uniform external electric field. \ Under some reasonable assumptions on the problem, we prove (see Theorem \ref {TeoremaPrincipale}) the existence of resonances with the associated metastable states at mean energy, in certain bands that contain the energy levels of the isolated molecule, too. \ Resonances are defined as complex eigenvalues of a distorted Hamiltonian; it is worth pointing out that our definition of resonances includes, as a special case, the notion of embedded eigenvalues where the imaginary part is exactly zero. \ In fact we prove that the imaginary part of the resonances is non positive and, in absolute value, smaller than $C h^2$, where $C$ is a given positive constant and $h$ is the semiclassical parameter defined below. \ Our result does not exclude the case of resonances with imaginary part exac
 tly zero (in fact, embedded eigenvalues). \ Finally, our result still holds true even in absence of the external field; in such a case we don't need to define the distorted Hamiltonian and we simply have discrete eigenvalues instead of resonances.

Because of the heavy mass of the nuclei (denoted by $h^{-2}$ in the text, with $0<h\ll 1$), at a first stage it is possible to consider the position of the nuclei as  fixed, in order to determine the electronic states. \ This approach is  known as  the Born-Oppenheimer's one \cite{KMSW}. \ Following Woolley, \cite{Wo}, we can say that the assumption of an almost fixed structure of the molecule is more appropriate for  metastable states, which is our case, than in the case of  definitely stationary states.

The first two levels of the electron, $E_1(R)$ of the state $1\sigma_g$ and $E_2(R)$ of the state $1\sigma_u$, as functions of the nuclear distance $R$, contribute to the effective potentials used for the determination of the nuclear dynamics. \ Such behavior of the electronic levels are well known by the explicit asymptotic expansions for large $R$ \cite{Silverstone} and their distributional Borel sums \cite{CGM,GG}. \
It is a reasonable hypothesis that each effective potential function, $W_j(R)=(1/R)+E_j(R),$ $j=1,2$,  has only one minimum point where a certain number of  nucleonic states are trapped, identified with  the first vibrational energy levels of the molecule. \ Such levels have been  experimentally observed \cite{CMM}. \ In the case of the coaxial Stark effect  we have  different behaviors of the electronic level functions $W_j(R)$,  see Fig. 2 of \cite{CMM}. \ The perturbation effect is an $R$ dependent splitting of the  asymptotically  degenerated levels, or shortly, an avoided crossing. \ The picture of the behavior of such effective potentials changes at large distances, giving a maximum of the first one $W_1(R)$. \ As a consequence, we have  metastability of the vibrational  states of the molecule.

Since we expect exponentially large mean life of the metastable states, we use a small complex distortion on the nuclear relative position variable \cite{Hu}.

We also use the theory of the twisted pseudodifferential operators introduced in \cite{MaSo2}. \ The theory of pseudodifferential operators goes back to the quantization rule of Hermann Weyl and is now  well established \cite{Ma1}. \ The recent theory of twisted pseudodifferential operators \cite{MaSo2} is a formalization and extension of the method of regularization going back to Hunziker \cite{Hu,KMSW}. \ This theory is able to regularize  the Coulomb singularity of the nuclei-electrons potentials of the interaction. \ The Grushin-Feshbach method is now a standard method for defining and computing a finite number of expected  eigenvalues \cite{KMSW}.

The Hamiltonian operator is the following:
\bee
H = - h^2 \Delta_{\RR}+ \frac {1}{R} + H_e   \label {Eq1}
\eee
where $h^2 \ll 1$ is the inverse of the nuclei mass (where we fix, for the sake of definiteness, the unit of mass such that the electron mass is equal to $1$) , and $H_e$ is the \emph {electronic Hamiltonian}, formally defined on $L^2 (\R^3_\rr )$ as
\bee
H_e := H_e (\RR ) = -  \Delta_{\rr} - \frac {1 }{|\rr - \frac 12 \RR |} - \frac {1 }{|\rr + \frac 12 \RR |} + V \, ,  \label {Eq2}
\eee
where $V$ is the external potential.

The three-body operator (\ref {Eq1}) acts on the Hilbert space  of square integrable sections in the trivial fiber bundle
\be
\K = L^2 (\R^3_\RR ; L^2( \R^3_\rr )) \, .
\ee
 In this picture the operator $H$ decomposes into two terms. \ The first one, the nuclear kinetic energy, acts on the base space. The second one operates on the fiber only,
 \be
 \tilde{H}_e=\int^\oplus H_e (\RR )d\RR\,,
 \ee
where $H_e (\RR )$ is the electronic Hamiltonian for fixed nuclei. \  The small parameter $h$ allows the use of semiclassical approximation. For our purposes, the  second order is enough.

Since $H_e$ is not simply a multiplication operator, we use the pseudodifferential calculus with operator valued symbols. \ For that, it is useful to translate the eigenvalue problem for $H$ by the Grushin method into the problem of inverting a $2\times 2$ matrix operator. \ Moreover we define the spectral projector $\Pi_e (\RR )$ of $H_e(\RR)$ up to a fixed value of the energy, so that,
\be
\Pi=\int^\oplus \Pi_e (\RR )d\RR\,,
\ee
is a projector on the molecular space $\K$. \ The lower part of the spectrum of the compressed operator $\Pi H\Pi$ is expected near of part of the spectrum of $H$. \ The eigenvalues are given by the generalized eigenvalues, $Q(E)\psi=E\psi,$ where $Q(E)$ is the  Feshbach operator,
\be
Q(E)=\Pi H\Pi-\Pi H (\Pi^\perp (H-E)^{-1})\Pi^\perp H \Pi.
\ee
Furthermore,  a smooth relationship, with respect to  $\RR$, is requested between the first eigenvectors of $H_e(\RR)$ and the corresponding final generalized eigenvectors of the Feshbach operator.

The paper is organized as follows.

In Section 2 we introduce the model and we state our main assumptions. \ In particular we discuss the $H_2^+$ molecule under the effect of a Stark effect and we show that this physical problem substantially fits with our model.

In Section 3 we consider the analytic distortion and regularization of the operator. \ Analytic distortion is a standard way to define resonances \cite {BCD}. \ Because of the singularity of the Coulomb potential we have to regularize our effective Hamiltonian. \ If we denote by $\tilde H_\mu$ ($\mu$ is the complex distortion constant) the regularized operator then we see (see Theorem 3.8) that part of its spectrum coincides with the spectrum of a reduced problem denoted by $\tilde P_\mu$. \ The reduced problem consists of two coupled Schr\"odinger operator.

In Section 4 we study the spectrum of the reduced problem denoted by $P^\sharp_\mu$, which coincides with $\tilde P_\mu$ up to a bounded operator with norm less that $C h^2$ for some $C>0$. \ We separately consider the spectrum associated to first level alone, and the part of the spectrum located in the bottom of the second level.

In Section 5 we compare the spectrum of the two operators $P^\sharp_\mu$ and $\tilde H^0_\mu$, where $\tilde H^0_\mu$ is the restriction of the regularized and distorted operator on the eigenspace of the vibrational spectrum.

In Section 6 we compare the spectrum of the two operators $\tilde H^0_\mu$ and $\tilde H^0_\mu$, where $H^0_\mu$ is the restriction of the distorted operator on the eigenspace of the vibrational spectrum.

In Section 7 we finally state our main results.

\subsection {Notations}

Here we list the main notations, meaning $j\in\{1,2\}$:

\begin {itemize}

\item [-] $H$ denotes the Hamiltonian operator (\ref {Eq1});

\item [-] $H_e$ denotes the electronic Hamiltonian operator (\ref {Eq2}) with eigenvalues ${\mathcal E}_j (R)$ depending on $R$;

\item [-] ${\mathcal H}_0 = Ker (\LL_\RR + \LL_\rr )$ where $\LL_\RR$ and $\LL_\rr$ respectively denote the angular momentum with respect to the variables $\RR$ and $\rr$;

\item [-] $W_j (R)= \frac 1R + {\mathcal E}_j (R)$ denotes the effective potential;

\item [-] $m_1$ and $m_2$ respectively are the non degenerate minima of $W_1 (R)$ and $W_2 (R)$ at $R_{1,m}$ and $R_{2,m}$, $M_1$ is the non degenerate maximum of $W_1 (R)$ at $R_{1,M}$ (see Remark \ref {Silverstone});

\item [-] $P_j$ is the operator formally defined by
\be
- h^2 \frac {d^2}{d R^2} + W_j (R)
\ee
on $L^2 (\R , d R)$ with  Dirichlet boundary conditions at $R=0$;

\item [-] ${\mathcal S}_\mu$ denotes the analytic distortion operator (\ref {eq16bis});

\item [-] $H_{\mu}$ and $H_{\mu ,e}$ denote the distorted operators
\be
H_{\mu} = {\mathcal S}_\mu H {\mathcal S}_\mu^{-1} \ \mbox { and } \ H_{\mu ,e} = {\mathcal S}_\mu H_e {\mathcal S}_\mu^{-1} \, ;
\ee

\item [-] $\widetilde H_{\mu , e}$ is the regularization of $H_{\mu ,e}$ as defined in Proposition \ref {qtilde};

\item [-] $\widetilde H_\mu$ is the regularization of $H_{\mu}$ as defined in Definition \ref {Hmumod};

\item [-] $H_\mu^0$ and $\widetilde H_\mu^0$ respectively are the restriction of $H_\mu$ and $\widetilde H_\mu$ to the invariant subspace $Ker (\LL_\RR + \LL_\rr )$;

\item [-] $P_\mu^\sharp$ is the reduced problem defined by equation (\ref {PmuBis}) on the Hilbert space
\be
\Hh^\sharp = L^2 \left ( [0 , + \infty ) , dR \right )\oplus L^2 \left ( [0 , + \infty ) , dR \right )
\ee
with  Dirichlet boundary conditions at $R=0$;

\item [-] $P_{j,\mu}$ is the operator formally defined by
\be
 h^2 {\mathcal S}_\mu \D_R^2 {\mathcal S}_\mu^{-1} + W_{j,\mu } (R)
\ee
on $L^2 (\R , d R)$ with  Dirichlet boundary conditions at $R=0$, where $\D_R$ and $W_{j,\mu }$ are defined at the beginning of \S 4;

\item [-] $P_D^\sharp$ is the Dirichlet realization of $P_0^\sharp$ on the interval $[0,R_{1,M} ]$;

\item [-] $\widetilde P_\mu^\sharp$ and $\widetilde P_j$ are respectively obtained by $P_\mu^\sharp$ and $P_j$ by substituting $\widetilde W_j$ to $W_j$, that is we "fill the well";

\item [-] $\widetilde P_\mu^0 (z)$ is the restriction of $\widetilde P_\mu (z)$, defined by equation (\ref {Pmu}), to $Ker (\LL_\RR )$;

\item [-] when this fact does not cause misunderstanding $\| \cdot \|$ denotes the usual norm on the Hilbert space $L^2$ or the norm of linear operators defined on the Hilbert space $L^2$.

\end {itemize}

\section {The model}

\subsection {The three-body problem}

The analysis of the three-body problem (\ref {Eq1}) is a very difficult task and we have to introduce here some suitable assumptions.

{\bf Hypothesis 1.} {\it We assume that the potential $V$ only depends on the component of the vector $\rr$ along the direction $\RR$; that is
\bee
V (\RR , \rr ) = \chi \left ( \left \langle \frac{\RR}{R},\rr \right \rangle \right )\, , \ R = |\RR |\, , \label {Eq2Bis}
\eee
where $\chi$ is a real-valued function bounded from below. \ The function $\chi$ admits an analytic extension in a complex strip containing the real axis.}

\begin {remark}\label{rem:O}\sl
Given a rotation $O$ in ${\R^3}$, let us consider the unitary operators $S_O$ and $ T_O$ on $L^2({\R}_{\RR}^{3})$ and $L^2({\R}_{\RR}^{3})\otimes L^2({\R_{\rr}}^{3}) $ respectively, given by,
\be
S_O\phi(\RR )=\phi(O \RR ), \quad \forall \phi\in L^2({\R}_{\RR}^{3})
\ee
\be
T_O:=S_O\otimes S_O,  \quad T_O \psi(\RR, \rr)=\psi(O\RR, O\rr), \quad \forall \psi\in L^2({\R}_{\RR}^{3})\otimes L^2({\R_{\rr}}^{3})
\ee
Then $H$ commutes with $ T_O$, i.e. $T_OH=HT_O$ and therefore, the spectrum of the electronic Hamiltonian operator $H_e (\RR) $ depends only on $R:=\vert \RR\vert$.
\end {remark}

\begin {remark} \label {INVH}\sl Now, let us denote by $\LL_\RR$ and $\LL_\rr$ the angular momentum with respect to  the variables $\RR$ and $\rr$ respectively. \ By the previous remark, we see that we have,
\bee
\label{invH}
[H, \LL_\RR +\LL_\rr ] =0.
\eee
In the sequel, we will be particularly interested on the eigenvalues and resonances of the restriction of $H$ to the invariant subspace
\be
{\mathcal H}_0 := Ker (\LL_\RR + \LL_\rr).
\ee%
This somehow corresponds to  fix to 0 the rotational energy of the molecule. As we will see, after the Born-Oppenheimer reduction to an effective Hamiltonian $P=P(\RR, hD_\RR)$, this is equivalent to study the restriction of $P$ to $Ker(\LL_\RR)$. Therefore, this will also permit us to reduce the study to a one-dimensional operator.
\end {remark}

\subsection {Effective Potential}

For any fixed $\RR\in{\R}^3$, we denote by ${\mbox {Sp}} \left ( H_e  (\RR) \right )$ the spectrum of the electronic Hamiltonian operator $H_e (\RR) $ defined on the Hilbert space $L^2 (\R^3_{\rr} )$. \ This spectrum actually depends on $R$ (see Remark \ref {rem:O}) and we assume that

{\bf Hypothesis 2.} {\it The discrete spectrum of the electronic Hamiltonian operator $H_e (\RR )$ contains at least two eigenvalues, and the first two eigenvalues ${\mathcal E}_1 (R )$ and ${\mathcal E}_2 (R )$ are non degenerate, extend holomorphically to complex values of $R$ in a domain of the form  $\Gamma_\delta:= \{ R\in\C\,;\, \re R \geq \delta^{-1},\, |\im R| < \delta \re R\}$ with $\delta >0$ constant, and are such that,
\bee
\lim_{{|R|\to + \infty},\, {R\in\Gamma_\delta} } {\mathcal E}_j (R) = {\mathcal E}_j^\infty ,
\eee
where,
\bee
{\mathcal E}_1^\infty < {\mathcal E}_2^\infty \, . \label {pippo}
\eee
Furthermore, there is a gap between ${\mathcal E}_j (R )$, $j=1,2$, and the remainder of the spectrum:
\be
\min_{R>0} \mbox {\rm dist} \left [ \left \{ {\mathcal E}_{1} (R), {\mathcal E}_{2} (R) \right \} , {\mathcal E}_3 (R)\right ] \ge C
\ee
for some positive constant $C>0$, where
\be
{\mathcal E}_3 (R)=  \left \{ \mbox {\rm Sp} (H_e (\RR)) - \left \{ {\mathcal E}_1 (R),\, \E_2(R) \right \} \right ] \} , .
\ee
}
We observe that, for any rotation $O$ in ${\R^3}$, one has (with obvious notations),
\be
H_e (O\RR , O\rr , O^{-1} D_\rr) =H_e (\RR, \rr, D_\rr)\, , \ \ D_\rr = -i \nabla_\rr .
\ee
As a consequence, the first two normalized eigenfunctions
\bee
H_e(\RR)\, \psi_{j} (\rr , \RR) = \E_j(R)\, \psi_j(\rr, \RR)\, , \ j=1,2\, . \label {Eq6}
\eee
can be taken real-valued and verify $\psi_j (O \RR,O \rr)= \psi_j (\RR ,\rr )$, and thus
\bee
(\LL_\RR + \LL_\rr)\psi_j =0.
\eee
We also denote by,
\be
W_j (R) =  \frac {1}{R} + {\mathcal E}_j (R ) \, ,\ j=1,2\, ,
\ee
the \emph {effective potential} associated with the $j$-th eigenvalue.

By Hypothesis 2 we observe that the effective potential satisfies to the following properties

\begin {enumerate}

\item The effective potentials $W_j (R)$, $j=1,2$, are analytic functions;

\item There exists a positive constant $C>0$ such that
\be
\min_{R>0} [ W_3 (R) - W_2 (R) ] \ge C
\ee
where $W_3 = \frac 1R + \mbox {\rm inf} \left [ {\mathcal E}_3 (R ) \right ]$.

\item The following limits hold true
\be
\lim_{R\to 0^+} W_j (R) = + \infty \, ,\ j=1,2 \, .
\ee

\end {enumerate}

Here, we introduce the following assumptions on the effective potentials $W_1 (R)$ and $W_2 (R)$.

{\bf Hypothesis 3.} {\it The effective potential $W_1$ has a single well shape, with local nondegenerate minimum value $m_1$ at some point $R_{1,m}$, with a barrier with local nondegenerate maximum value $M_1$ at some point $R_{1,M}$; beside, $W_1$ does not admit other critical points in the domain $W^{-1}_1 \left ( [ m_1 , M_1 ] \right )$. \ The effective potential $W_2$ has a single well shape, with local minimum value $m_2$ at some point $R_{2,m}>R(1,m)$.}

\begin {remark} \label {Silverstone}\sl
In absence of the external field  the local maximum value $M_1$ disappears and we only have two local minimum values \cite {Silverstone}, in such a case ${\mathcal E}_1^\infty = {\mathcal E}_2^\infty$ and we could treat the spectral problem for eigenvalues belonging to the interval $[m_1 , \widetilde M_1 ]$, for any $\widetilde M_1 < {\mathcal E}_1^\infty$. \ If the external field is small enough, but not zero, then we expect to observe a local maximum value such that $m_1 < m_2 < M_1$ and $R_{1,m} < R_{2,m} < R_{1,M}$ as in Fig. \ref {Fig1}. \ For increasing external field, as considered by \cite {MS}, can happen  to have  $m_1 <M_1 < m_2$.
\end {remark}

\begin{figure}
\begin{center}
\includegraphics[height=8cm,width=8cm]{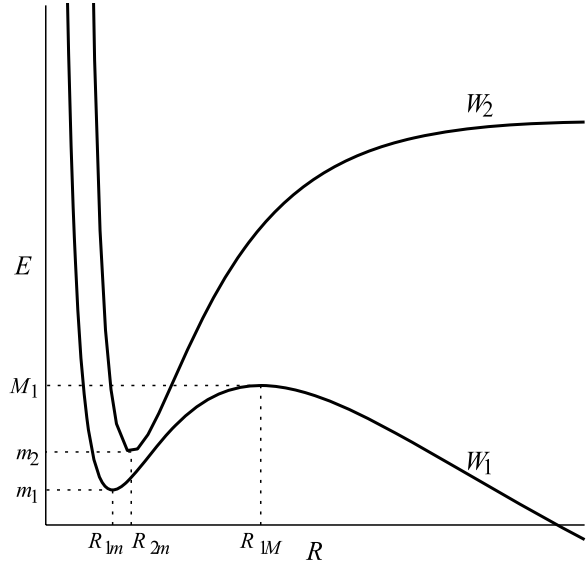}
\caption{Graph of the effective potentials $W_1 (R)$ and $W_2 (R)$ with single well shapes. \ The effective potential $W_1 (R)$ has a barrier and it does not admit other critical points in the domain $W^{-1}_1 \left ( [ m_1 , M_1 ] \right )$; where $m_1$ and $M_1$ are the values of the local maximum and minimum point of $W_1$.}
\label {Fig1}
\end{center}
\end{figure}

\subsection {Spectrum of the reduced operator}

In polar coordinates, Hamiltonian (\ref {Eq1}) takes the form,
\bee
H = - h^2 \left [ \frac {\partial^2}{\partial R^2} + \frac {2}{R} \frac {\partial}{\partial R} \right ] - h^2 \frac {1}{R^2} \Lambda^2  + \frac 1R + H_e (\RR ) \label {Equa8Bis}
\eee
where $\Lambda^2$ is the Legendrian operator,
\be
\Lambda^2 = \frac {1}{\sin \theta } \frac {\partial}{\partial \theta} {\sin \theta } \frac {\partial}{\partial \theta} + \frac {1}{\sin^2 \theta} \frac {\partial^2}{\partial \varphi^2}.
\ee
The operator $- h^2 \frac {1}{R^2} \Lambda^2$ has eigenvalues $h^2 \frac {1}{R^2} \ell (\ell +1)$, $\ell \in \{0,1,2,\ldots\}$. \ As a consequence, using Remark \ref {INVH}, a suitable choice of the rotation $O$ makes the operator $H$ take the form,
\be
H = - h^2 \left [ \frac {\partial^2}{\partial R^2} + \frac {2}{R} \frac {\partial}{\partial R} \right ] + h^2 \frac {\ell (\ell+1)}{R^2}  + \frac 1R + H_e (R )
\ee
on $L^2( \R_+, R^2 dR; L^2 (\R^3_{\rr} ) )$. \ Finally, by taking $\ell =0$, that is, by considering the restriction of $H$ on $Ker (\LL_\RR )$ (still denoted by $H$), and by performing the change $\psi (R , \rr ) \to R \psi (R , \rr)$, the Hamiltonian $H$ takes the form,
\be
H_0 = - h^2 \frac {\partial^2}{\partial R^2} + \frac 1R + H_e (R )
\ee
on $L^2 (\R_+ ,dR; L^2( \R^3_\rr ))$ with Dirichlet boundary condition at $R=0$.

Let $P_j$, $j=1,2$, be the \emph {reduced operator} formally defined by
\bee
\label{defPj0}
P_j = -h^2 \frac {d^2}{d R^2}+ W_j (R), \ \ W_j (R) = \frac {1}{R} + {\mathcal E}_j (R) ,
\eee
on the Hilbert space $L^2 (\R_{+},dR ) $ with Dirichlet boundary condition at $R=0$.

Then, it follows that for $h $ small enough and for small external field, the discrete spectra of $P_j$ in the interval $[m_j , {\mathcal E}_j^\infty )$, $j=1,2$, is not empty (see, e.g., \cite {Landau} in the case without the external field), and we denote it by
\be
{\mbox {Sp}}_{\text {d}} \left ( P_j \right ) = \left \{ e_k^j ,\ k\geq 1 \right \} \, , \ j=1,2\, .
\ee
In particular, in the case of non degenerate minima points $m_1$ and $m_2$, combining results from \cite{HeRo} and \cite{HeSj1}, we know that the gap $e_{k+1}^j-e_k^j$ between two consecutive eigenvalues of $P_j$ ($j=1,2$) is of order $h$ as $h\rightarrow 0_+$, in the sense that $c_{j}h\leq e_{k+1}^j-e_k^j\leq C_{j}h$ with $c_{j}, C_{j} >0$ independent of $h$ and $k =\Oo (h^{-1})$.

\subsection {Physical motivation}
It is well known \cite {CMM} that the dynamics of the three particle system called molecule-ion $H_2^+$, referred to its center of mass, and under the effect of an external homogeneous field, is described by  a Hamiltonian operator of the form,
\bee
H = - h^2 \Delta_{\RR}+ \frac {Z^2}{R} + H_e   \label {Eq1BIS}
\eee
where $h \ll 1$ is the effective semiclassical parameter given by the square root of the ratio between the light mass of the electron $e$ and the heavy mass (when compared with the electron mass) of the hydrogen nuclei. \ Moreover,
\be
\RR = x \x + y \y + z \z
\ee
is the relative position  of  the two hydrogen nuclei of $H_2^+$, and $Z$ is the electron charge (hereafter, for the sake of definiteness we assume that the units are such that, $Z =1$). \ The electronic Hamiltonian (\ref {Eq2}) describes the relative motion of the electron $e$ referred to the \emph {fixed} nuclei, and it actually depends on the nuclear distance $R$, where $V$ is the  potential of the external force.

Now we consider the isolated molecule. The asymptotic behavior for large $R$ of the functions $W_j(R)$ is dominated by the Van der Waals force  given by,
$$
W_j(R)=-\frac{c_4}{R^4}+O(R^{-6}),\,\, W_j'(R)=4\frac{c_4}{R^5}+O(R^{-7}),
$$ 
for a constant $c_4>0$ (see the constant $E^{(4)}$ of \cite {Silverstone}). \ The energy binding  of the molecule, $\mathcal{E}_1^\infty-m_1>0$ is much smaller than the separation distance  of the fundamental level of the atom $\mathcal{E}_3^\infty-\mathcal{E}_1^\infty>0$.

Following \cite {Hiskes, MS}, we consider the case where the external field is a Stark-like field directed along the axes of the two nuclei, in agreement with Hyp. 1, with potential
\bee
V (\RR , \rr ) = - \chi \left ( \left \langle \frac{\RR}{R},\rr \right \rangle \right )\, , \ R = |\RR |\, ,
\eee
where $\chi$ is a real-valued function bounded from below.  \ In such a case, the Hamiltonian $H$ commutes with the angular momentum $\LL_\RR + \LL_\rr$. \ In fact, the effective potentials $W_1 (R)$ and $W_2 (R)$ show a shape as in Fig. \ref {Fig1} (see, e.g., \cite {MS}, and   Fig. 2 of \cite{CMM}).

The function $\chi$ admits an analytic extension to a complex strip containing the real axis. \ More precisely, we  assume, $$\chi(x)=\chi_d(x)= \nu\frac{x}{\sqrt{1+(x/d)^2}}=\frac{d\nu }{\sqrt {1+(d/x)^2}}=d\nu(1-\frac{d^2}{2x^2}+\frac{3d^4}{8x^4}+O(\frac{d^6}{x^6})),$$  for $x^2>d^2,$ where $d>0$ is a parameter much larger than the molecular size.

For 
$$\Delta \mathcal{E}(R)=\mathcal{E}_2(R)-\mathcal{E}_1(R),$$ 
small enough for $R^2>>d^2$, we can approximate with a degenerate perturbation  with effective Stark potential for the nucleons given by the matrix element of the potential on the two electronic states $\psi_{j,R}=\psi_{j,R}(r)$,
\bee
V_S (R)=|<\psi_{1,R},V (\RR , . )\psi_{2,R}>|=C\chi(R) =\delta(1-\frac{d^2}{2R^2}+\frac{3d^4}{8R^4}+O(\frac{d^6}{R^6})) ,
\eee
where $\delta =C d\nu>0$.
 Thus, for large $R$, the two nuclear potentials are bounded and behaves as in Fig. 1,
$$
W_1(R)= \mathcal{E}_1^\infty-\delta(1-\frac{d^2}{2R^2})-
(c_4+\delta\frac{3d^4}{8})\frac{1}{R^4}+O(1/R^6),$$$$ W_2(R)= \mathcal{E}_1^\infty+\delta(1-\frac{d^2}{2R^2})-
(c_4-\delta\frac{3d^4}{8})\frac{1}{R^4}+O(1/R^6).
$$
If we admit derivation of the asymptotic series, we have a maximum $R_{1,M}$ of  $
W_1(R)$ diverging as $\delta\rightarrow 0$, $$R_{1,M}^2\sim \frac{4c_4}{\delta d^2}+\frac{3d^2}{2}.$$
\section {Analytic distortion and regularization of the operator}

\subsection {Analytic distortion}

Let  $s \in C^{\infty}({\R})$,  $0\leq s \leq 1$ with $ s (x)=0$ in  an arbitrarily large compact set containing 0, and $ s (x)=1$ if $|x|$ is large enough. \ For $\mu$  real small enough, we set,
\begin{eqnarray}
\label{Imu}
&& I_\mu:{\R}^3\rightarrow {\R}^3, \quad I_\mu(\RR)=\RR(1+\mu s(R)) \\
\label{Jmu}
&& J_\mu:{\R}^6\rightarrow {\R}^3, \quad J_\mu(\RR,\rr)=\rr\left [ 1+\mu s\left ( \left \langle \frac{\RR}{R},\rr \right \rangle \right ) \right ],
\end{eqnarray}
and we define the analytic distortion on the test function $\varphi$, by the formula,
\bee
\left ( \Ss_\mu \varphi \right ) (\RR, \rr)= \vert J(\RR, \rr)\vert^{1/2} \varphi (I_\mu(\RR), J_\mu(\RR,\rr)), \label {eq16bis}
\eee
where we have set $R=|\RR|$, and $J(\RR, \rr)$ is the Jacobian of the transformation $F_\mu$ given by,
\begin{eqnarray}
&& F_\mu:{\R}^6\rightarrow {\R}^6, F_\mu(\RR,\rr)=(I_\mu(\RR), J_\mu(\RR,\rr)).
\end{eqnarray}
We also set
\be
\phi_\mu:{\R}_+\rightarrow {\R}_+, \phi_\mu(R)=R(1+\mu s(R)).
\ee
Then, the analytic distortion applied to the operator (\ref {Eq1}) defined on the Hilbert space $\K$ takes the form,
\bee
H_{\mu}=\Ss_\mu H \Ss_\mu ^{-1}=-h^2 \Ss_\mu \Delta_{\RR} \Ss_\mu^{-1}+\frac{1} {\phi_\mu(R)}+H_{\mu , e} (\RR)\, , \label {Eq14}
\eee
with $H_{\mu , e} (\RR)$ given by,
\be
H_{\mu , e} (\RR) = -  \Ss_\mu \Delta_{\rr} \Ss_\mu^{-1} - \frac {1 }{| J_\mu(\RR,\rr) - \frac 12 I_\mu(\RR)|} -
\frac {1 }{| J_\mu(\RR, \rr) + \frac 12  I_\mu(\RR) |} + V^{\mu},
\ee
where the distorted external potential is given by,
\be
V_\mu (\RR,\rr) = V \left [ \left \langle \frac{\RR}{R},\rr \right \rangle \left (1+\mu s\left ( \left \langle \frac{\RR}{R},\rr \right \rangle \right ) \right ) \right ].
\ee
Thus, $H_{\mu , e} (\RR)$ can be extended to small enough complex values of $\mu$ as an analytic family of type A.

\begin {remark} \sl
We also observe that, if $O$ is a rotation in $\R^3$, then,
\be
I_\mu (O\RR) = O I_\mu (\RR)\quad ;\quad J_\mu (O\RR, O\rr) = OJ_\mu (\RR,\rr)\, .
\ee
As a consequence,
\bee
\label{invrotS}
[ \Ss_\mu, \LL_\RR +\LL_\rr] = 0,
\eee
and,
\be
H_{\mu ,e} (O\RR , O\rr , O^{-1} D_\rr) =H_{\mu ,e} (\RR , \rr, D_\rr ).
\ee
We denote by $H_{\mu ,0}$ the restriction of $H_{\mu ,e} (\RR , \rr, D_\rr )$ to the invariant subspace $Ker (\LL_\RR + \LL_\rr)$.
\end {remark}

\subsection {Regularization of $H_{\mu}$}

In this section, we want to regularize the operator $H_{\mu}$ with respect to the $\RR$-variable. \ Having in mind the representation (\ref {Equa8Bis}) of the Laplacian in polar coordinates, we denote
\be
\Omega(1/M):=\left\{\RR\in{\R}^3\; : \;R>\frac{1}{M}\right\}, \ \ \Omega_0 (1/M):=\left\{\RR\in{\R}^3\; : \;R<\frac{1}{M}\right\} \, ,
\ee
and $S^2$ is the unit sphere in $\R^3_\RR$, and
\be
L_0:=-\Delta_\rr+C_0
\ee
with $C_0, M>0$ large enough. \ We have the following preliminary technical lemma:

\begin{lemma}\sl
\label{Lem2}
Under the previous assumptions, there exists a finite family of conical open sets $(\Omega_\ell)_{\ell=1}^m$ in ${\R}^3$, of the form $\Omega_\ell=]\frac{1}{M}, +\infty[ \times \omega_\ell$ with $\omega_\ell$ bounded open set of $S^{2}$, and a corresponding family of unitary operators $\U_\ell(\RR)$ ($\ell=1,\cdots ,m$, $\RR\in \Omega_\ell$) on $L^2 (\R^3_\rr )$, such that (denoting by $U_\ell$ the unitary operator on $L^2(\Omega_\ell;L^2({\R}^{3}_\rr) ) \simeq L^2(\Omega_\ell)\otimes{L^2({\R}^{3}_\rr})$ induced by the action of $\U_\ell(\RR)$ on $L^2({\R}^{3}_\rr)$), one has,

\begin{enumerate}

\item $\Omega(1/M)=\mathop{\displaystyle\cup_{\ell=1}^m\Omega_\ell}$;

\item For all $\ell=1,\cdots, m$ and $\RR\in  \Omega_\ell$, $\U_\ell(\RR)$ leaves $H^2({\R}^{3}_\rr)$ invariant;

\item For all $\ell$, the operator $\U_\ell(-h^2\Ss_\mu\Delta_\RR\Ss_\mu^{-1})\U_\ell^{-1}$  is a semiclassical differential operator with operator-valued symbols, of the form,
\begin{equation}
\label{conjomega}
-h^2\Ss_\mu\Delta_\RR\Ss_\mu^{-1}+h\sum_{|\beta|= 1}\omega_{\beta ,\ell}(\RR)(hD_{\RR})^\beta+
h^2\omega_{0 ,\ell}(\RR) \, , \ D_\RR = -i \nabla_\RR
\end{equation}
where $\omega_{\beta ,\ell}L_0^{\frac{|\beta|}{2} -1}\in C^\infty (\Omega_\ell; {\mathcal  L}(L^2({\R}^{3}_\rr))$, and,
for any $\gamma\in\N^3$, the quantity  $\Vert\partial_x^\gamma \omega_{\beta ,\ell}(x) L_0^{\frac{|\beta|}{2} -1}\Vert_{{\mathcal  L}(L^2({\R}^{3}_\rr))} $ is bounded uniformly with respect to $h$ small enough and locally uniformly with respect to $x\in\Omega_\ell$;

\item For all $\ell$, the operators $\U_\ell H_{\mu , e}\U_\ell^{-1}$  are in $C^\infty (\Omega_\ell; {\mathcal  L}
(H^2({\R}^{3}_\rr), L^2({\R}^{3}_\rr))$.

\end{enumerate}
\end{lemma}

\begin{proof}
At first, let us make  a  change of variables  as in \cite{MaMe}, that localizes into a compact set  the  $\RR$ - dependent singularities appearing into the interaction potential. \ Let $\chi\in C^{\infty}({\R}^+)$ satisfying $0\leq \chi\leq 1, \chi'\leq 0$, such that,
$$
\chi(s)=1, \;\;\;\;\mbox{if}\;\;  0\leq s\leq 1, \quad\quad \chi(s)=0, \;\;\mbox{if}\;\; s\geq 2
$$
For $\tau>1/2M$ and  $\tau>0$, we consider the function,
$$
\rho(\tau,t)=\frac{t}{\tau}\chi \left ( \frac{t}{\tau} \right )+2Mt \left ( 1-\chi \left ( \frac{t}{\tau} \right ) \right ) .
$$
Then, it is easy to check that
\begin{eqnarray*}
&&\frac{\partial \rho}{\partial t}>0 \;\;\mbox{on } \left ] \frac{1}{2M}, +\infty \right [ \times {\R}_+,\\
&&\rho \;\;\mbox{is surjective onto }{\R}_+, \\
&&\frac{\partial^k \rho}{\partial \tau^k} \;\;\mbox{is uniformly bounded on } ]\frac{1}{2M}, +\infty[\times {\R}_+, \forall k\geq 1.
\end{eqnarray*}
Therefore we can define $\alpha_\tau$ as the inverse diffeomorphism on ${\R}_+$ of the function $t\rightarrow \rho (\tau,t)$. \ in particular, by construction we have,
$$
\alpha_\tau(t)=\frac{t}{2M}\;\;\mbox{if}\;\; t\geq 4M\tau,\quad \alpha_\tau(t)=\tau t \;\;\mbox{if} \;\;t\leq 1.
$$
Now, for $\RR\in \Omega (1/M)$, we define
$$
\theta(\RR,.):{\R}^3\rightarrow {\R}^3, \quad \theta(\RR, \rr)=\alpha_{R/2}(\vert \rr\vert)\frac{\rr}{\vert \rr\vert}.
$$
Then, for any $\RR\in \Omega(1/M)$, the function $ \theta(\RR,.)$ is a diffeomorphism of ${\R}^3$, it depends smoothly on $\RR$, and is such that $\partial_\RR^{\alpha}\theta(\RR,\rr)$ is uniformly bounded on $ \Omega(1/2M)\times {\R}^3$, for any $\alpha\in{\N}^3\setminus\{0\}$ (see \cite{MaMe}, Lemma 3.1). \ Moreover
\begin{eqnarray*}
&& \theta \left ( \RR,\frac{\RR}{R} \right ) =\frac{\RR}{2},\\
&& \theta(\RR,\rr)=\frac{\rr}{2M}\; \;\;\mbox {for } \vert \rr\vert\geq 2MR,\\
&& \theta(\RR,\rr)=\frac{R }{2}\rr \;\;\;\mbox {for } \vert \rr\vert\leq 1.
\end{eqnarray*}
For $R >\frac{1}{M}$,  we consider the unitary transformation on $L^2({\R}^3_\rr)$, given by,
$$
( U(\RR)\phi)(\rr)=\phi(\theta(\RR,\rr)) \vert\partial_{\rr}\theta(\RR,\rr)\vert^{1/2}\, .
$$
The advantage of performing this change of variables is that the $\RR$-depending singularities of the potential are now localized in some compact subset of ${\R}^3_\rr$. \ Now, following the arguments of \cite{KMSW}  and with the help of the previous change of variables, let us show that, by a patch and cut procedure, one can localize the singularities of the potential at some fixed ($\RR$-independent) points.

For any fixed $ z_0\in S^{2}$ (the unit sphere in ${\R}^3$), we choose a functions $f_{z_0} \in C_0^\infty (\R^3; \R)$, such that,
$$
f_{z_0}(z_0) = 1,\quad\quad f_{z_0}(-z_0)=0
$$
and, for $z$ close enough to $z_0$ and $s\in\R^3$, we define
$$
F_{z_0}(z,s) := s + (z-z_0)(f_{z_0}(s)-f_{z_0}(-s))\in \R^3.\\
$$
For $z$ in a smooth neighborhood $\omega_{z_0}$ of $z_0$, the application $s\mapsto F_{z_0}(z,s)$ is a diffeomorphism of $\R^{3}$, and we have,
$$
F_{z_0}(z, z_0)=z, \quad F_{z_0}(z, -z_0)=-z.
$$
Moreover, for any $\alpha\in {\R}^3$, there exists $C_{\alpha}>0$ such that, for any $z\in\omega_{z_0}$, for any $s,s'\in {\R}^3$
\begin{eqnarray*}
&& \frac{1}{C_0}\vert s -s'\vert\leq \vert F_{z_0}(z,s) -F_{z_0}(z,s') \vert\leq  C_0\vert s -s'\vert\\
&&\vert \partial_x^{\alpha}F_{z_0}(z,s) -\partial_x^{\alpha}F_{z_0}(z,s')\vert\leq C_{\alpha}\vert s -s'\vert\\
&&\vert \partial_x^{\alpha}F_{z_0}(z,s)\vert\leq C_0,\quad \vert \alpha\vert\geq 1
\end{eqnarray*}
If $(\omega_\ell)_{\ell=1}^m:= (\omega_{z_\ell})_{\ell=1}^m$ is a family of  such open sets that covers  $S^2$, we set $F_\ell(z, .):=F_{z_\ell}(z,.)$, and we define,
$$\Omega_{\ell}:=\left ]\frac{1}{M}, +\infty \right [ \times \omega_{\ell}.
$$
For $\RR\in\Omega_{\ell}$, we also set,
$$(\widetilde U_\ell(\RR)\phi)(\rr)=\left \vert {\rm det} (\partial_\rr F_\ell ) \left ( \frac{\RR}{R}, \rr \right ) \right \vert^{1/2}\phi \left (F_\ell \left (\frac{\RR}{R},\rr \right )\right ),
$$
and,
$$\U_\ell(\RR):=\widetilde U_\ell(\RR)U(\RR);$$
$$(\U_\ell(\RR)\phi)(\rr))=\phi( \gamma_\ell(\RR,\rr))\vert {\rm det}(\partial_\rr \gamma_\ell)(\RR,\rr)\vert,$$
where
$$\gamma_\ell(\RR,\rr)=\theta \left ( \RR, F_\ell \left ( \frac{\RR}{R},\rr \right ) \right ).
$$
Then, it  is easy to check (see \cite{MaMe}) that  $ \U_\ell$ satisfy (1), (2), (3), and (4). This completes the proof of the lemma.
\end{proof}

Now, let us consider the spectral projection $\Pi_0(\RR)$ associated to $\{{\mathcal E}_1(R), {\mathcal E}_2(R)\}$ of $H_e(\RR)$, where $\E_1 (R)$ and $\E_2 (R)$ are the first two (simple) eigenvalues of $H_e(\RR)$. \ If one denote by $\gamma(R)$ a continuous simple loop in $\C$ enclosing $\{{\mathcal E}_1(R), {\mathcal E}_2(R)\}$ and having the rest of ${\mbox {Sp}}\left (H_{e} (\RR)\right )$ in its exterior,  one can write $\Pi_0(\RR)$ as,
\be
\Pi_0(\RR) =\frac{1}{2\pi i}\int_{\gamma(R)}(H_{e} (\RR)-z)^{-1}\;dz .
\ee
Moreover, for $\mu $ complex small enough, one can define the projector,
\be
\Pi_{\mu ,0} (\RR)=\frac{1}{2\pi i}\int_{\gamma(R)}(H_{\mu , e} (\RR)-z)^{-1}\;dz
\ee
satisfying $(\Pi_{\mu ,0})^*=\Pi_{\bar\mu ,0}$. \ We have the following

\begin{lemma} \label{Lem3}\sl
There exist two functions,
\be
w_1^{\mu}(\RR,\rr), w_2^{\mu}(\RR,\rr)\in C^0({\R}_\RR^3;H^2({\R}^3_\rr))
\ee
depending analytically on $\mu$, and real-valued for $\mu$ real, such that,

\begin{itemize}

\item [i.] $\langle w_k^{\mu}(\RR,\rr), w_{l}^{\bar \mu}(\RR,\rr)\rangle_{L^2({\R}^3_\rr)}=\delta_{k,l},
\quad k,l=1,2$;

\item [ii.] $w_j^{\mu} \in C^\infty(\Omega_0 (2/M) ;H^2({\R}^3_\rr))$, $j=1,2$, and, for $\RR \in \Omega (3/M)$, $w_1^{\mu}(\RR,\rr)$ and $w_2^{\mu}(\RR,\rr)$ form a basis of ${\rm Ran}\Pi_0^{\mu}$;

\item [iii.] For $\RR \in \Omega (3/M)$, $w_1^{\mu}(\RR,\rr) $ and $ w_2^{\mu}(\RR,\rr)$ are eigenfunctions of $H_{\mu , e} (\RR)$ associated to ${\mathcal E}_1(\phi_\mu(R))$ and $ {\mathcal E}_2( \phi_\mu(R))$ respectively;

\item [iv.] For all $\ell=1,\dots, m$, one has $\U_\ell(\RR)w_j^{\mu}(\RR,\rr) \in C_b^{\infty}(\Omega_\ell(M), H^2({\R}^3_\rr))$, $j=1,2$.

\item [v.] $w_1^{\mu}$ and $w_2^{\mu}$ can be chosen in such a way that,

$(\LL_\RR +\LL_\rr)w_1^{\mu}=(\LL_\RR +\LL_\rr)w_2^{\mu}=0$.

\end{itemize}

\end{lemma}

\begin {proof} Taking into account that (see (\ref {pippo})),
\be
\lim_{R\mapsto +\infty}{\mathcal E}_1(R)\not=\lim_{R\mapsto +\infty}{\mathcal E}_2(R),
\ee
the points (i)-(iv) follow from  Lemma 3.1 of \cite{MaSo2}  and  from  the arguments of  Proposition 5.1 in \cite{MaMe}. Moreover, since ${\mathcal E}_1(R)$ and ${\mathcal E}_2(R)$ are  non degenerate, the last point (v) follows from \cite{KMSW}, Theorem 2.1, and from (\ref{invrotS}).
\end {proof}

\vskip 0.3cm\noindent

Thanks to the previous lemma, we see that the family $(    U_\ell,\Omega_\ell)_{\ell=0,m}$ (with $\Omega_0 = \Omega_0 (2/M),  \U_0=\I $ and $ \Omega_\ell, \U_\ell$  defined  in Lemma \ref{Lem2} and Lemma \ref{Lem3}), is a regular unitary covering of $L^2 ({\R}^3_\RR; L^2({\R}^3_\rr))$ in the sense of  \cite{MaSo2}, Definition 4.1.

We set,
\be
\widetilde \Pi_{\mu ,0}  (\RR) = \langle \cdot ,w_1^{\mu}(\RR)\rangle_{L^2(\R_{\rr}^{3})}w_1^{\mu}(\RR)+
\langle \cdot ,w_2^{\mu}(\RR)\rangle_{L^2(\R_{\rr}^{3})}w_2^{\mu}(\RR)
\ee
so that  $\widetilde \Pi_{\mu ,0} (R)$ coincides with $\Pi_{\mu ,0} (R)$ for $\RR \in \Omega (3/M)$, and  verify,
\be
\U_\ell(\RR)\widetilde \Pi_{\mu ,0} (\RR) \U_\ell(\RR)^{-1}\in  C^\infty (\Omega_\ell,{\mathcal  L}(L^2(\R_{\rr}^{3}))),
\ee
for all $\ell=0,\dots, m$. \ Also observe that, for any rotation $O$,
$$T_O\widetilde \Pi_{\mu ,0} (\RR)=\widetilde \Pi_{\mu ,0} (\RR)T_O$$
or, equivalently, $[\Pi_{\mu ,0} (\RR), \LL_\RR+\LL_\rr]=0$.

We also denote by  $ \widetilde\Pi_0(\RR)$ the value of $ \widetilde\Pi_{\mu ,0}(\RR)$ for $\mu=0$.

Now, with the help of  $ \widetilde \Pi_{\mu ,0}(\RR)$,  we modify $H_{\mu , e} (\RR)$ outside a neighborhood of $\Omega (5/M)$ as follows (see Proposition 3.2 in \cite{MaSo2}).

\begin{proposition} \sl
\label{qtilde}
We choose a function $\zeta\in  C^{\infty}(\R_+ ; [0,1])$, such that $\zeta=1$ for  $R\geq 3/M$ and $\mbox {\rm supp} \zeta \subseteq ]2/M , + \infty [$. \ Then, for all $\RR\in \R^3$, and $\mu$ complex small enough, there exists an operator $\widetilde H_{\mu , e} (\RR)$ on $L^2(\R_\rr^3)$, with domain $H^2(\R_\rr^3)$,  depending analytically on $\mu$, such that,
\be
&& \widetilde H_{\mu , e} (\RR)=H_{\mu , e} (\RR)\quad \;\;\mbox{if}\,\,  \RR \in \Omega (4/M) ;\\
&&[\widetilde H_{\mu , e} (\RR), \widetilde\Pi_{\mu ,0}(\RR)] =0\quad \;\;\mbox{for all } \RR \in {\R}^3 ,
\ee
and the application $\RR\mapsto \U_\ell(\RR)\widetilde H_{\mu , e} (\RR) \U_\ell(\RR)^{-1}$ is in $C^\infty (\Omega_\ell; {\mathcal L}(H^2(\R_{\rr}^3), L^2(\R_{\rr}^{3}))$  for all  $\ell=0,\dots, m$. \ Moreover, $\widetilde H_{\mu , e} (\RR)$ commutes with $\LL_\RR +\LL_\rr$, in the sense that, for any $\varphi\in C_0^\infty (\R^6)$, one has,
$$
(\LL_\RR +\LL_\rr)\widetilde H_{\mu , e} (\RR)\varphi = \widetilde H_{\mu , e} (\RR)(\LL_\RR +\LL_\rr)\varphi.
$$
Hence, the spectrum of $\widetilde H_{\mu , e} (\RR)$ actually depends only on $R\in \R_+$. Moreover, for $\mu$ real, $\widetilde H_{\mu , e} (\RR)$ is selfadjoint, uniformly semibounded from below, and the bottom of its spectrum consists in two  eigenvalues,
\be
\widetilde {\mathcal E}_j^\mu (R) = \widetilde {\mathcal E}_j (\phi_\mu(R))\, ,\  j=1,2 ,
\ee
where
\be
\widetilde {\mathcal E}_j (R) = \zeta (R) {\mathcal E}_j (R)\, .
\ee
Furthermore, $\widetilde H_{\mu , e} (\RR)$ admits a global gap in its spectrum, in the sense that,
\be
\inf_{R \in \R_+}{\mbox {\rm dist}}( \{\widetilde {\mathcal E}_1^\mu(R), \widetilde {\mathcal E}_2^\mu(R)\}, \widetilde {\mathcal E}_3^\mu(R)) >0.
\ee
where we set
\be
\widetilde {\mathcal E}_3^\mu(R) = {\mbox {\rm Sp}} (\widetilde H_{\mu ,e} (\RR))\backslash \{\widetilde {\mathcal E}_1^\mu (R), \widetilde {\mathcal E}_2^\mu (R)\}
\ee
\end{proposition}

\begin {proof} The proof is similar to that of Proposition 3.2 in \cite{MaSo2}, and we write it for $\mu=0$ only (the general case is obtained by just substituting $H_{\mu ,e}$ to $H_e$ and $\widetilde\Pi_{\mu ,0}$ to $\widetilde\Pi_0$). \ We set $\widetilde\Pi_0^{\perp}(\RR)=1-\widetilde\Pi_0(\RR)$ and
\be
\widetilde H_{e} (\RR)=\zeta(R) H_{e} (\RR)+(1-\zeta(R))\widetilde\Pi_0^{\perp}(\RR) (-\Delta_\rr+C_0) \widetilde \Pi_0^{\perp}(\RR) .
\ee
with $C_0>0$ large enough and such that $C_0>{\bar{\mathcal E}_3}$, where
\bee
{\bar {\mathcal E}_3}:=\inf_R {\mathcal E}_3(R) \, . \label {Eq15}
\eee
Since $\widetilde\Pi_0(\RR)=\Pi_0(\RR)$ on ${\rm Supp}\hskip 1pt \zeta$, we see that  $\widetilde\Pi_0(\RR)$ commutes with $\widetilde H_{e} (\RR)$, and it is also clear that $\widetilde H_{e} (\RR)$ is selfadjoint with domain $H^2({\R}^{3})$. \ Moreover,
\be
\widetilde\Pi_0(\RR)H_{e} (\RR)\widetilde\Pi_0(\RR)= \zeta (\RR)\Pi_0(\RR) H_{e} (\RR)\Pi_0(\RR),
\ee
and,
\begin{equation}
\label{reducedop}
\widetilde \Pi_0^{\perp}(\RR)\widetilde H_{e} (\RR)\widetilde \Pi_0^{\perp}(\RR)\geq \left(\zeta (R) {\mathcal E}_3(R)+ (1-\zeta (R))C_0\right)\widetilde\Pi_0^{\perp}(\RR)\geq {\bar {\mathcal E}_3}\widetilde \Pi_0^{\perp}(\RR).
\end{equation}
In particular, the bottom of the spectrum of $\widetilde H_{e}(\RR)$ consists in two eigenvalues $\widetilde {\mathcal E}_j (R) = \zeta (R) {\mathcal E}_j (R)$ with associated eigenvectors $\widetilde \Pi_0 (\RR) \psi_j $, $j=1,2$, where $\E_j$ and $\psi_j$ are the first two eigenvalues and eigenvectors of (\ref {Eq6}). \ Furthermore, one has
\be
&&\inf_{R>2/M}{\rm dist }(\widetilde {\mathcal E}_3(R), \{\widetilde {\mathcal E}_1(R), \widetilde {\mathcal E}_2(R)\})\\
&&\hskip 3cm \geq \inf_{R>2/M} \left ( \zeta (R)( \mbox { inf } \left [ {\mathcal E}_3(R) \right ] -{\mathcal E}_2(R)) + (1-\zeta (R))C_0 \right ) >0,
\ee
and
\be
\inf_{0<R\leq2/M}{\rm dist}(\widetilde {\mathcal E}_3(R), \{\widetilde {\mathcal E}_1(R), \widetilde {\mathcal E}_2(R)\})\geq C_0.
\ee
In particular, $\widetilde H_{e} (\RR)$ admits a fix global gap in its spectrum as stated in the proposition. \ Finally, we see that $\widetilde H_{e} (\RR)$ commutes with $\LL_\RR +\LL_\rr$, and $\U_\ell\widetilde H_{e} (\RR)\U_\ell^{-1}$ depends smoothly on $\RR$ in $\Omega_\ell$ for all $\ell=0,\dots, m$. \end {proof}

\subsection{Regularization of the operator}

\begin {definition} [Regularization of $H_{\mu}$] \sl
\label{Hmumod}
Let $\Ss_\mu$ be the analytic distortion defined in (\ref {eq16bis}) for $\mu$ in some small complex neighborhood of zero, and let $\widetilde H_{\mu ,e} (\RR )$ and $\zeta (R)$ be defined as in Proposition \ref {qtilde}.  Then, we define the regularization of $H_{\mu}$ as,
\bee
\label {Eq16}
\widetilde H_{\mu}= -h ^2 \Ss_\mu \Delta_\RR \Ss_\mu^{-1}+\widetilde H_{\mu , e} (\RR)+\frac{\zeta(R)} {\phi_\mu(R)}+\frac{M}{3}(1-\zeta (R)).
\eee
\end {definition}

Taking into account Definition 4.4 in \cite{MaSo2}, we see that Lemma \ref{Lem2}, Proposition \ref{qtilde} and (\ref{invrotS}) imply,

\begin{lemma} \label{Lem4}\sl The operator $\widetilde H_{\mu}$ is a twisted
PDO (of degree 2) on $L^2(\R_\RR^3,L^2(\R_\rr^3))$ (in the sense of Definition 5.1 in \cite{MaSo2}),
 associated with the regular unitary covering $(\U_\ell,\Omega_\ell)_{\ell=0,\dots, m}$. Moreover, it commutes with $\LL_\RR +\LL_\rr$.
\label{lem:Utwisted2}
\end{lemma}

Now, we define
\be
Z_{\mu}^+ : L^2(\R^6)  \rightarrow  L^2(\R_\RR^3) \oplus L^2(\R_\RR^3)\ee
by the formula,
\be
\left ( Z_\mu^+\psi \right ) (\RR)=\langle \psi (\RR , \cdot ), w_1^{\bar\mu}(\RR , \cdot )\rangle_{L^2(\R_{\rr}^{3})}\oplus \langle \psi (\RR ,\cdot ),  w_2^{\bar \mu}(\RR ,\cdot )\rangle_{L^2({\R}_{\rr}^{3})},
\ee
and,
\be
Z_{\mu}^-=(Z_{\bar \mu}^+)^* :L^2(\R_\RR^3) \oplus L^2(\R_\RR^3) \rightarrow L^2(\R^6),
\ee
by
\be
\left ( Z_\mu^-(u_1\oplus u_2)\right ) (\RR,\rr)=u_1(\RR) w_1^{\mu}(\RR,\rr)+u_2(\RR) w_2^{ \mu}(\RR,\rr).
\ee

Following \cite{MaMe}, we consider the Grushin problem,
 \be
\widetilde {\mathcal G}_\mu(z) =\left( \begin{array}{cc} \widetilde H_{\mu}-z & Z_{ \mu}^-\\
Z_\mu^+ & 0  \end{array} \right),
\ee
that sends $H^2(\R^6)\oplus \left(L^2(\R^3) \oplus L^2(\R^3) \right)$ into $L^2(\R^6) \oplus \left( H^2(\R^3)\oplus H^2(\R^3 )\right)$.

Thanks to Lemma \ref{Lem3} and Lemma \ref{Lem4}, we see that $\widetilde {\mathcal G}_\mu(z)$ is a twisted PDO (of degree 2) on $L^2\left(\R^3_\RR ;  L^2(\R^3_\rr) \oplus \C\oplus\C \right)$,  associated with the regular unitary covering $({\mathcal V}_\ell,\Omega_\ell)_{\ell=0,\dots, m}$, where we have set
$$
{\mathcal V}_\ell := \left( \begin{array}{cc} \U_\ell & 0\\
0 & 1_2  \end{array} \right).
$$
We also have,

\begin{lemma} \sl
\label{invrotG}
For all $\mu\in\C$ small enough and $z\in\C$, the operator
$\widetilde {\mathcal G}_\mu(z)$ commutes with  ${\mathcal L}:= \left( \begin{array}{ccc} L_\RR +L_\rr & 0 \\
0 &   L_\RR \end{array} \right)$.
\end{lemma}

\begin {proof} By Lemma \ref{Lem4}, we only need to study the commutation rules between $Z_\mu^\pm$ and the operators $L_\RR$ and $L_\rr$. But, using Lemma \ref{Lem3}, v., plus the fact that the formal adjoint of $L_\rr$ is $-L_\rr$, we  immediately obtain,
$$
(L_\RR +L_\rr)Z_\mu^- = Z_\mu^- L_\RR\quad ; \quad L_\RR Z_\mu^+ = Z_\mu^+(L_\RR +L_\rr),
$$
and the result follows.
\end{proof}

Moreover, we see as in \cite{MaMe}, Section 5, that, for $z\in\C$ with ${\rm Re}z < \inf_R \widetilde {\E}_3 (R)$ and ${\rm Im}z$ sufficiently small, the operator $\widetilde {\mathcal G}_\mu(z)$ is invertible, and its inverse $\widetilde {\mathcal G}_\mu(z)^{-1}$ is such that the operators,
$$
 \left( \begin{array}{cc} 1 & 0\\
0 & \langle -\Delta_\RR\rangle^{-1}  \end{array} \right)\widetilde {\mathcal G}_\mu(z)^{-1}, \quad \widetilde {\mathcal G}_\mu(z)^{-1} \left( \begin{array}{cc} 1 & 0\\
0 & \langle -\Delta_\RR\rangle^{-1}  \end{array} \right)
$$
are twisted (bounded) $h$-admissible operators associated with the regular unitary covering $({\mathcal V}_\ell,\Omega_\ell)_{\ell=0,\dots, m}$. As a consequence, $\widetilde {\mathcal G}_\mu(z)^{-1}$ can be written as,
$$
\widetilde {\mathcal G}_\mu(z)^{-1} =\left( \begin{array}{cc} E_\mu (z) & E_\mu^+ (z)\\
E_\mu^- (z) & z-\widetilde P_\mu (z) \end{array} \right),
$$
where $\widetilde P_\mu(z)$ is an unbounded $h$-admissible operator on $L^2(\R_\RR^3) \oplus L^2(\R_\RR^3)$ with domain $H^2(\R_\RR^3)\oplus H^2(\R_\RR^3)$, and $E_\mu (z)$, $E_\mu^\pm (z)$ are (bounded) twisted $h$-admissible operators (all depending in a holomorphic way on $z$).
 \vskip 0.2cm
More precisely, it results from \cite{MaMe}, formula (2.11), that the operator $\widetilde P_\mu(z)$ is given by the formula,
\begin{equation}
\label{defopeff}
\widetilde P_\mu(z) = Z_\mu^+\widetilde H_{\mu}Z_\mu^- - Z_\mu^+[ h ^2 \Ss_\mu \Delta_\RR \Ss_\mu^{-1},\widetilde \Pi_{\mu ,0}](\widetilde H_{\mu}'-z)^{-1}[ \widetilde \Pi_{\mu ,0}, h ^2 \Ss_\mu \Delta_\RR \Ss_\mu^{-1}]Z_\mu^-,
\end{equation}
where $\widetilde \Pi_{\mu ,0}$ stands for the projection on $L^2({\R}^6)$ induced by the action of $\widetilde \Pi_{\mu ,0}(\RR)$ on $L^2(\R_{\rr}^{3})$, and $\widetilde H_{\mu}'$ is the restriction of $(1-\widetilde \Pi_{\mu ,0})\widetilde H_{\mu}(1-\widetilde \Pi_{\mu ,0})$ to the range of $1-\widetilde \Pi_{\mu ,0}$. \ In particular, $\widetilde H_{\mu}'-z$ is invertible in virtue of (\ref{reducedop}), and $\widetilde \Pi_{\mu ,0}$ is a twisted $h$-admissible operator on $L^2(\R_\RR^3,L^2(\R_\rr^3))$ (in the sense of Definition 4.4 in \cite{MaSo2}), associated with the regular unitary covering $(\U_\ell,\Omega_\ell)_{\ell=0,\dots, m}$.
 \vskip 0.2cm
By Lemma \ref{invrotG}, we also know that $\widetilde P_\mu (z)$ commutes with $L_\RR$, and thus, gathering all the previous information on $\widetilde P_\mu (z)$, we finally obtain that it can be written as,
\begin{eqnarray}
\label{Pmu}
\widetilde P_\mu(z) =  -h ^2 \Ss_\mu \Delta_\RR \Ss_\mu^{-1} +{\mathcal M}_\mu(R)+h {\mathcal A}_\mu (\RR,h \D_\RR)+ h ^2 {\mathcal B}_\mu(\RR,h \D_\RR ; z,h )
\end{eqnarray}
where, for any  $R>\frac{3}{M}$, ${\mathcal M}_\mu$ is given by,
\begin{eqnarray}
\label{Mmu}
{\mathcal M}_\mu(R)=
 \left( \begin{array}{cc}
W_1 ( \phi_\mu(R)) & 0 \\ 0 & W_2( \phi_\mu(R))
\end{array} \right) \, ,
W_j (R) = {\mathcal E}_j (R)+\frac{1}{R} \, , \label {eq27bis}
\end{eqnarray}
and,  for $R\leq \frac{3}{M}$ and $\mu$ sufficiently small, it satisfies,
\begin{eqnarray}
\label{Mmu1}
\re {\mathcal M}_\mu (R) \geq \frac{M}{4} + \inf_R {\mathcal E}_1(R). \label {StimaBasso}
\end{eqnarray}
Here, $M$  is the same as in Proposition \ref{qtilde} and Definition \ref{Hmumod}, and it can be chosen arbitrarily large.

Moreover ${\mathcal A}_\mu (\RR ,h\D_\RR)$ is of the form,
\begin{eqnarray}
{\mathcal A}_\mu = \left(\begin{array}{cc} 0 & a_\mu(\RR) \cdot h \D_\RR \\ h \D_\RR \cdot \overline {a_{\mu} }(\RR)  &  0
\end{array} \right) ,  \label {eq27ter}
\end{eqnarray}
for some smooth bounded (together with its derivatives) function $a_\mu(\RR)$ independent of $z$. Finally,  ${\mathcal B}_\mu(\RR,h \D_\RR; z, h )$ is an $h $-admissible pseudodifferential operator depending analytically on $z$, with Weyl symbol $b_\mu(\RR,\RR^* ; z, h )$ holomorphic with respect to $R^*$ in a complex strip of the form $\{ |\Im R^*| < \delta\}$ (with $\delta >0$ independent of $z$ and $\mu$), such that, for any multi-index $\alpha$,
\begin{eqnarray}
\label{estbmu}
 \partial^\alpha b_\mu(\RR,\RR^* ;z,   h )  =\Oo (1)
\end{eqnarray}
uniformly with respect to $(\RR,\RR^*)\in \R^3 \times \R^3$, $h >0$ small enough, and $z$ close enough to some fix $\lambda_0\in\C$ such that ${\rm Re}\lambda_0 < \inf_R \widetilde {\E}_3 (R)$ and ${\rm Im}\lambda_0$ sufficiently small.
 \vskip 0.2cm
Finally, the operators ${\mathcal A}_\mu$ and ${\mathcal B}_\mu(\RR,h \D_\RR; z, h )$ commute with $L_\RR$, and one has the Feshbach identities,
\begin{eqnarray}
\label{feschbachRes}
&&(\widetilde H_\mu -z)^{-1} = E_\mu (z)+E_\mu^+(z)(\widetilde P_\mu(z)-z)^{-1}E_\mu^-(z),\\
&& (\widetilde P_\mu(z)-z)^{-1} =Z_\mu^+ (\widetilde H_\mu -z)^{-1}Z_\mu^-\nonumber.
\end{eqnarray}

Summing up, we have proved,

\begin{theorem} \label {Thm3}\sl
Let $\widetilde {\E}_3 (R)$ be defined as in Proposition \ref {qtilde} and let $\lambda_0\in \C$ with ${\rm Re}( \lambda_0) <
\inf_R \widetilde {\E}_3 (R)$ and ${\rm Im}( \lambda_0)$ sufficiently small. Under the previous assumptions, there exists a complex neighborhood $D_{\lambda_0}$ of $\lambda_0$ such that, for any $z\in D_{\lambda_0}$, one has the equivalence,
\be
z\in {\rm  Sp} (\widetilde H_{\mu})\Longleftrightarrow z\in {\rm  Sp} (\widetilde P_{\mu}(z)),
\ee
where $\widetilde P_{\mu}(z)$is as in (\ref{Pmu}) with (\ref{Mmu})-(\ref{estbmu}).
\end{theorem}

Now, taking advantage of Lemma \ref{invrotG}, we can consider the restriction of the Grushin problem $\widetilde {\mathcal G}_\mu(z)$ on $Ker (\LL_\RR +\LL_\rr)\oplus Ker (\LL_\RR )\oplus Ker (\LL_\RR )$, and we also immediately obtain,

\begin{corollary}\sl
\label{CorThm3}
Denote by $\widetilde H_\mu^0$ the restriction of $\widetilde H_{\mu}$ on the invariant subspace $Ker (\LL_\RR +\LL_\rr)$, and by $\widetilde P_{\mu}^0(z)$ the restriction of $\widetilde P_{\mu}(z)$ on the invariant subspace $Ker (\LL_\RR )$. Then, for any $\lambda_0\in \C$ with ${\rm Re}( \lambda_0) <\inf_R \widetilde {\E}_3 (R)$ and ${\rm Im}( \lambda_0)$ sufficiently small, there exists a complex neighborhood $D_{\lambda_0}$ of $\lambda_0$ such that, for any $z\in D_{\lambda_0}$, one has the equivalence,
\be
z\in {\rm  Sp} (\widetilde H_\mu^0)\Longleftrightarrow z\in {\rm  Sp} (\widetilde P_\mu^0(z)).
\ee
\end{corollary}

In the sequel, we also denote by $H_\mu^0$ the restriction of $ H_{\mu}$ on the invariant subspace $Ker (\LL_\RR +\LL_\rr)$.

\section{Reduced problem}

Let us introduce the following shortcut notation,
\be
\D = h \D_R \, , \ \D_R = - i \frac {d }{d R}.
\ee
For all $R>0$, we define,
$$
W_{j,\mu }(R):= \zeta (R)W_j(\phi_\mu (R))+ \frac{M}3 (1-\zeta (R)),
$$
where $\zeta$ is as in Proposition  \ref{qtilde}, and $M$ is taken large enough. In particular, $W_{j,\mu }$ is bounded and depends analytically on $\mu$.
 \vskip 0.2cm
Then, we set,
$${\mathcal M}_\mu^0 (R):=  \left( \begin{array}{cc}
W_{\mu , 1}(R) & 0 \\ 0 & W_{\mu , 2}(R)
\end{array} \right),
$$
and we denote by $ {\mathcal A}_\mu^0 (R,h \D_R)$ the restriction of the differential operator (actually, vector-field) $ {\mathcal A}_\mu (\RR,h \D_\RR)$ on the space $Ker (\LL_\RR )$. In particular, since $[{\mathcal A}_\mu , \LL_\RR ]=0$ then ${\mathcal A}_\mu^0 $ can be represented as a differential operator in the variable $R=|\RR|$, and it can be written as,
$$
{\mathcal A}_\mu^0 ={\mathcal A}_\mu^0 (R,h \D_R)=\left(\begin{array}{cc} 0 & a_\mu^0(R) h \D_R \\  h \D_R \cdot\overline {a_\mu^0 }(R) &  0
\end{array} \right) ,
$$
where the function $a_\mu^0$ is smooth and bounded together with all its derivatives on $(0,+\infty)$.
\vskip 0.2cm
In this section, we look for the solutions to the eigenvalue equation,
\bee \label{Eq17}
P^\sharp_\mu \varphi =\lambda\varphi,\quad \lambda\in\C, \, \, \varphi = \left(\begin{array}{c}
\varphi_1 \\ \varphi_2 \end{array} \right),
\eee
where $P^\sharp_\mu$ is the differential operator formally defined as,
\bee
P^\sharp_\mu =   h ^2 \Ss_\mu \D_R^2 \Ss_\mu^{-1} +{\mathcal M}_\mu^0 (R)+h {\mathcal A}_\mu^0 (R,h \D_R) \label{PmuBis}
\eee
acting on the Hilbert space,
\bee
\Hh^\sharp = L^2 \left ( [0 , + \infty ) , dR \right )\oplus L^2 \left ( [0 , + \infty ) , dR \right ), \label {Eq32}
\eee
with zero Dirichlet boundary condition at $R=0$, and where now, with abuse of notation, we denote,
\bee
\label{abusenot}
\left ( \Ss_\mu \varphi \right ) (R) = |I' (R) |^{1/2} \varphi [I(R)] \, , \ I(R) = R \left [ 1 + \mu s(R) \right ] \, .
\eee

If we set,
\be
P_{j,\mu} = h^2 \Ss_\mu \D_R^2 \Ss_\mu^{-1} + W_{j,\mu }(R) \, ,
\ee
then equation (\ref{Eq17}) turns into
\bee
(P_{1,\mu} -\lambda)\, \varphi_1 & = -h A_\mu^0 \varphi_2 ; \label {Eq18} \\
(P_{2,\mu} -\lambda)\, \varphi_2 & = -h {A_\mu^0}^* \varphi_1 , \label {Eq19}
\eee
where
\be
A_\mu^0 = h a_\mu^0 (R) \D_R \, .
\ee
Let us also observe that, by the Weyl theorem, the essential spectrum of $P_\mu^\sharp$ is given by,
$$
{\rm Sp}_{ess} (P^\sharp_\mu) = {\mathcal E}_1^\infty+(1+\mu)^{-2}[ 0,+\infty).
$$

As before, $m_j$ is the local minima of $ W_j$ and $M^1$ is the local maximum of $W^1$, as defined in Remark \ref {Silverstone}. \ Now, for the sake of definiteness, we consider the case where
\bee
{\mathcal E}^\infty_1 < m_1 < m_2 < M_1
\eee
In the case where ${\mathcal E}^\infty_1 < m_1 < M_1 < m_2 $ then we can apply the same argument to the interval $[m_1, M_1]$.

\begin {remark} \sl \label {Minimo}
By construction, we have ${\rm Sp} (H_{\mu ,e} (\RR) )={\rm Sp} (H_e(I_\mu (\RR))$. Therefore, if the function $s(x)$ used in the distortion vanishes in a sufficiently compact set (and since $W_{j,\mu }$ and $W_j$ coincide on this set), a continuity argument shows that, for $\mu\in\C$ small enough,  the critical points of $\re W_{j,\mu }$  and $W_j$ coincide and remain non-degenerate.
\end {remark}

As a consequence, for any $\lambda \in [m_1,m_2+\alpha]$ (with $\alpha >0$ small enough), the function $W_1(R) -\lambda$ presents the shape of a {\it well in an island} in the sense of \cite{HeSj2}. Moreover, since we are in dimension one, the complementary of  the island (that is, the non compact component of $\{W_1 \leq \lambda\}$) is automatically non-trapping, and we can adopt the general strategy used in \cite{Ma2} (see also \cite{CMR, FLM}), that consists in taking $\mu =2ih\ln\frac1{h}$ in the definition of the analytic distortion. The function $s$ used in (\ref{Imu})-(\ref{Jmu}) can also be assumed to be 0 on a neighborhood of the ``greatest'' island, defined by $\{ R>0\,; W_1(R)\geq m_1\}$. Then, following \cite{FLM}, Theorem 2.2 (see also \cite{HeSj2}, Proposition 9.6), we first show that the eigenvalues of $P_\mu^\sharp$ with their real part in $[m_1,m_2+\alpha]$, coincide, up to an exponentially small error term, with  eigenvalues of the Dirichlet realization $P^\sharp_D$ of $P^\sharp_0$ on the interval $[0,R_{1,M}]$.

\begin{proposition}\sl
\label{compPmuPD}Let $\alpha >0$ small enough, and let ${\mathcal J}\subset (0,1]$, with $0\in\overline{\mathcal J}$, such that there exists a function $a(h)>0$ defined for $h\in {\mathcal J}$ and verifying,
\begin{eqnarray}
&& \mbox{For all } \varepsilon >0,\, a(h)\geq \frac1{C_\varepsilon}e^{-\varepsilon /h} \mbox{ for $h\in {\mathcal J}$ small enough}; \\
&& {\rm Sp} (P^\sharp_D)\cap [m_2+\alpha-2a(h), m_2+\alpha+2a(h)] =\emptyset.
\end{eqnarray}
Set,
$$
\Omega (h) := \{ z\in\C\,; \, {\rm dist}(\re z, [m_1,m_2+\alpha])<a(h), \, |\im z| < C^{-1}h\ln\frac1{h}\},
$$
with $C>0$ a large enough constant. Then, there exists $\delta_0>0$ and a bijection,
$$
b\, :\, {\rm Sp} (P^\sharp_D) \cap [m_1,m_2+\alpha] \rightarrow {\rm Sp} (P_\mu^\sharp)\cap \Omega (h),
$$
such that,
$$
b(\lambda )-\lambda ={\mathcal O}(e^{-\delta_0/h}),
$$
uniformly for $h\in {\mathcal J}$.
\end{proposition}
\begin{remark}\sl
\label{remecarts}
In our situation, it is well known (see, e.g., \cite{HeRo}) that the distance between two consecutive eigenvalues of the Dirichlet realizations of $P_{1}$ and $P_{2}$ on $(0,R_{1,M})$, behaves like $h$ as $h\rightarrow 0_+$. Then, by slightly moving the parameter $\alpha$, it is not difficult to deduce that  the previous proposition actually gives a complete description of the spectrum of $P_\mu^\sharp$ in a  neighborhood of $[m_1, m_2+\alpha]$, for all sufficiently small values of $h>0$.
\end{remark}
\begin{proof}
At first, we fix a function $F =F(R)\in C_0^\infty ((0,R_{1,M}) ;\R_+)$, such that,
$$
\inf_{(0,R_{1,M}]}(W_1+F) > m_2+\alpha,
$$
and we denote by $p_{j,\mu}=p_{j,\mu}(R,R^*)$ the principal symbol of the operator $P_{j,\mu}$. We also denote by $\widetilde p_{j,\mu}$ an almost analytic extension of $p_{j,\mu}$ (see, e.g., \cite{MeSj}).
\vskip 0.2cm
Then, using the fact that the whole interval of energy $[m_1, m_2+\alpha]$ is non-trapping for the operator $P_1+F$, we can construct as in \cite{CMR} Section 7 (or \cite{Ma2} Section 4), a real valued function $f_0 =f_0(R,R^*)\in C_0^\infty((\R_+ \backslash \mbox{Supp }F)\times \R)$, such that, on the set $\{F(R) +\re  p_{1,\mu} (R,R^*) \in [m_1-\delta, m_2+\alpha+\delta]\}$ (with $\delta >0$ small enough), one has,
$$
-\im \widetilde p_{1,\mu} \left ( R-h\ln\frac1{h} \left ( \partial_R f_0 + i \partial_{R^*}f_0 \right ) , R^* -h\ln\frac1{h} \left ( \partial_{R^*}f_0 - i \partial_R f_0 \right ) \right ) \geq \delta h\ln\frac1{h}.
$$
As a consequence (see, e.g., \cite{CMR} Section 7), if $z\in\C$ is such that ${\rm dist }(z, [m_1, m_2+\alpha]) << h\ln (1/h)$, then, the operator $ P_{1,\mu}+F -z$ is invertible on $L^2 (\R_+)$, and its inverse satisfies,
\bee
\Vert h^{-f_0}T(P_{1,\mu}+F-z)^{-1}u\Vert_{L^2(\R^2)} \leq C|h\ln h|^{-1}\Vert h^{-f_0}Tu\Vert_{L^2(\R^2)}, \label {Paperino}
\eee
where $C>0$ is a constant, and $T\,:\, L^2(\R_+)\rightarrow L^2(\R^2)$ is the Bargmann tranform, defined by,
$$
Tu(R,R^*):=\frac1{2\pi h}\int_{R' >0} e^{i(R-R')R^*/h -(R-R')^2/2h}u(R')dR'.
$$
Indeed, for any $v \in C_0^\infty (\R_+ )$ we have,
\bee
\Vert h^{-f_0}T  v\Vert_{L^2(\R^2)} \leq \frac {C}{|h\ln h|} \Vert h^{-f_0}T (P_{1,\mu}+F-z) v\Vert_{L^2(\R^2)} \, , \label {pippobis}
\eee
and, by means of a density argument, we can extend such an estimate to any $v \in H^2 \cap H^1_0 (\R_+ )$. \ Then, inequality (\ref {Paperino}) holds true for the function $u = (P_{1,\mu}+F-z) v$, which belongs to the space $L^2 (\R_+)$ and satisfies the Dirichlet condition at $R=0$.

This means that the operator $(P_{1,\mu}+F-z)^{-1}$ has a norm ${\mathcal O}(|h\ln h|^{-1})$ if we consider it as acting on the space ${\mathcal H}=L^2(\R_+ )$ endowed with the norm given by: $\Vert u\Vert_{\mathcal H}:= \Vert h^{-f_0}Tu\Vert_{L^2(\R^2)}$.
\vskip 0.2cm
On the other hand, by construction, the operator $ P_{2,\mu} + F$ has a real part greater than $m_2+\alpha$, and thus, if $\re z\leq m_2+\alpha$, we also see that the operator $(P_{2,\mu}+F-z)^{-1}$ has a uniformly bounded norm when acting on ${\mathcal H}$.
\vskip 0.2cm
Then, proceeding as in \cite{HeSj2}, Section 9 (see also \cite{FLM}, Section 2), we pick up two functions $\chi_1, \chi_2\in C_0^\infty ((0,R_{1,M}); [0,1])$, such that $\chi_1 =1$ in a neighborhood of Supp${\chi_2}$, and $\chi_2 =1$ in a neighborhood of Supp$F$. Setting,
\bee
\label{resolvapp}
Q_\mu^\sharp := P_\mu^\sharp+F \quad ;\quad R_\mu^\sharp (z):= \chi_1(P^\sharp_D-z)^{-1}\chi_2 + (Q_\mu^\sharp-z)^{-1}(1-\chi_2),
\eee
we see that, if ${\rm dist} (z,{\rm Sp} (P^\sharp_D))\geq a(h)$, then (\cite{HeSj2}, Formula (9.39) and Proposition 9.8),
$$
(P_\mu^\sharp-z)R_\mu^\sharp (z) = I+ K_\mu(z) \mbox{ with }\Vert K_\mu(z)\Vert_{{\mathcal L}(\mathcal H )} =\Oo (e^{-2\delta /h}),
$$
where $\delta >0$ is some constant. Therefore, for such values of $z$ and for $h$ small enough, we have,
\bee
\label{resolvNeumann}
(P_\mu^\sharp-z)^{-1} = R_\mu^\sharp (z)\sum_{j\geq 0}(-K_\mu(z))^j,
\eee
and since $\Vert R_\mu^\sharp (z)\Vert_{\mathcal H} =\Oo (h^{-C})$ for some constant $C>0$, we deduce that, if $\gamma$ is a simple oriented loop around ${\rm Sp} (P^\sharp_D) \cap [m_1,m_2+\alpha]$ such that ${\rm dist}(\gamma , {\rm Sp} (P^\sharp_D))\geq a(h)$ and ${\rm dist}(\gamma , [m_1,m_2+\alpha])<< |h\ln h|$, then,
\begin{eqnarray}
\Pi_\mu^\sharp := \frac1{2i\pi}\int_{\gamma}(z-P_\mu^\sharp)^{-1}dz &=& -\frac1{2i\pi}\int_{\gamma}R_\mu^\sharp (z) +\Oo (e^{-\delta /h})\nonumber\\
\label{compDir}
&=& \frac1{2i\pi}\int_{\gamma} \chi_1(z-P^\sharp_D)^{-1}\chi_2dz+\Oo (e^{-\delta /h}).
\end{eqnarray}
Here, we have also used the fact that $z\mapsto (Q_\mu^\sharp-z)^{-1}$ is holomorphic in the interior  of $\gamma$, that can be taken equal to $\Omega(h)$.

Now, since $\Pi_\mu^\sharp$ is the spectral projector of $P_\mu^\sharp$ associated with $\Omega(h)$, the corresponding resonances of $P^\sharp$ are nothing but the eigenvalues of $P_\mu^\sharp\Pi_\mu^\sharp$ restricted to the range of $\Pi_\mu^\sharp$. Moreover, if we set $\{ \mu_1, \dots, \mu_m\}:= {\rm Sp} (P^\sharp_D)\cap [m_1,m_2+\alpha]$, and if we denote by $\varphi_1, \dots , \varphi_m$ an orthonormal basis of $\bigoplus_{j=1}^m {\rm Ker} (P^\sharp_D-\mu_j)$, then, by  Agmon estimates, we see on (\ref{compDir}) (see also \cite{HeSj2},Theorem 9.9 and Corollary 9.10) that the functions $\Pi_\mu^\sharp \chi_1 \varphi_j$ ($j=1,\dots,m$) form a basis of Ran$\Pi_\mu^\sharp$, and the matrix of $P_\mu^\sharp\left\vert_{{\rm Ran}\Pi_\mu^\sharp}\right.$ in this basis, is of the form ${\rm diag} (\mu_1, \dots, \mu_m) +\Oo (e^{-\delta /h})$. Then, the result follows from standard arguments on the eigenvalues of finite matrices (plus the fact that $m=\Oo (h^{-N_0})$ for some $N_0\geq 1$ constant).
\end{proof}

Now, exploiting the fact that both $W_1(R_{1,M})$ and $W_2(R_{1,M})$ are (strictly) greater than $m_2$, we consider two functions $ \widetilde W_j \in C^\infty (\R_+ ;\R)$ ($j=1,2$), such that,
$$
\widetilde W_j =W_j \mbox{ on } [0, R_{1,M}]\, ;\, \widetilde W_j \mbox{ is constant on } [2R_1^M, +\infty) \, ;\, \	 \inf_{[R_{1,M}, +\infty)}\widetilde W_j > m_2,
$$
and we set,
$$\widetilde{\mathcal M}_0 (R):=  \left( \begin{array}{cc}
\widetilde W_1(R) & 0 \\ 0 & \widetilde W_2(R)
\end{array} \right),
$$
$$
\widetilde P^\sharp := h ^2  \D_R^2  +\widetilde{\mathcal M}_0 (R)+h {\mathcal A}_0 (R,h \D_R),
$$
acting on the space ${\mathcal H}^\sharp$.
That is, $\widetilde P^\sharp$ is obtained form $P^\sharp$ by substituting $\widetilde W_1, \widetilde W_2$ to $W_1, W_2$. Then, the same arguments used in Proposition \ref{compPmuPD} (and actually simpler, since both operators are self-adjoints)  show that, under the same conditions, the spectrum of $P^\sharp_D$ and the spectrum of $\widetilde P^\sharp$ coincide in $[m_1, m_1+\alpha +a(h)]$, up to some exponentially small error-terms. Therefore, in order to know the resonances of $P^\sharp$ in $\Omega (h)$ (up to those exponentially small error-terms), it is sufficient to study the eigenvalues $\lambda$ of the self-adjoint $\widetilde P^\sharp$ in $[m_1, m_1+\alpha +a(h)]$.

For $j=1,2$, we set,
$$
\widetilde P_j := h^2 D_R^2 +\widetilde W_j,
$$
acting on $L^2(\R^+; dR)$ with Dirichlet condition at $R=0$, and we consider separately two different cases.

\subsection {Case 1: $\lambda \le (m_2-\alpha)$}\label{sub511} In that case, the operator $\widetilde P_2-\lambda$ is invertible, with a uniformly bounded inverse, and the equation,
\bee \label{Eqtilde}
\widetilde P^\sharp \varphi =\lambda\varphi,\quad  \varphi = \left(\begin{array}{c}
\varphi_1 \\ \varphi_2 \end{array} \right),
\eee
can be re-written as,
\begin{eqnarray*}
&& \varphi_2 = -h (\widetilde P_2 -\lambda)^{-1} A_0^*\varphi_1;\\
&& \left [ \widetilde P_1 - h ^2 A_0 ( \widetilde P_2 -\lambda)^{-1} A_0^* \right ]  \varphi_1 = \lambda \varphi_1.
\end{eqnarray*}
Thus,  the eigenvalues $\lambda$ are given by the equation,
\bee
\label{Eq22}
\lambda= \widetilde{f}_k (\lambda),
\eee
where the $\widetilde{f}_k (\lambda)$'s are the eigenvalues of $\hat P_1(\lambda ):= \widetilde P_1 - h ^2 A_0 (\widetilde P_2 -\lambda)^{-1} A_0^*$.
\vskip 0.2cm
Writing,
$$
\hat P_1(\lambda )-z = (1-h^2A_0(\widetilde P_2-\lambda)^{-1}A_0^*(\widetilde P_1-z)^{-1})(\widetilde P_1-z),
$$

and observing that, for $z\notin \mbox{\rm Sp}(\widetilde P_1)$, $A_0^*(\widetilde P_1-z)^{-1}$ is bounded and has a norm $\Oo (\mbox{\rm dist } (z, \mbox{\rm Sp}( \widetilde P_1))^{-1})$, we conclude that, if $\mbox{\rm dist } (z, \mbox{\rm Sp}( \widetilde P_1))>> h^2$, then $\hat P_1(\lambda )-z$ is invertible, and its inverse satisfies,
$$
(\hat P_1(\lambda )-z)^{-1} = ( \widetilde P_1-z)^{-1} (1+\Oo (h^2/\mbox{\rm dist } (z, \mbox{\rm Sp}( \widetilde P_1)))).
$$
Differentiating with respect to $\lambda$, we also obtain,
$$
\frac{d}{d\lambda}(\hat P_1(\lambda )-z)^{-1} =( \widetilde P_1-z)^{-1} \Oo (h^2/\mbox{\rm dist } (z, \mbox{\rm Sp}( \widetilde P_1))) =\Oo (h^2/\mbox{\rm dist } (z, \mbox{\rm Sp}( \widetilde P_1))^2).
$$
Then, using the fact that, under the non degenerate condition discussed in Remark \ref {Minimo}, the eigenvalues $E_{1,k}$ ($k\geq 1$) of $ \widetilde P_1$ are distant at least  of order $h$ between each other, for each of them we can define the projection,
$$
\hat \Pi_1(\lambda) := \frac1{2i\pi}\int_{\gamma_k}(z-\hat P_1(\lambda ))^{-1} dz,
$$
where $\gamma_k$ is a complex oriented simple circle  centered at $E_k^1$ of radius $\delta h$ with $\delta >0$ small enough. Applying standard regular perturbation theory, we easily conclude that the $k$-th eigenvalue of $\widetilde{f}_k (\lambda)$ of $\hat P_1(\lambda)$ satisfies,
\bee
\label{estfk}
f_k(\lambda) = E_{1,k} +\Oo (h^2)\quad ;\quad \frac{df_k}{d\lambda}(\lambda) =\Oo (h),
\eee
uniformly with respect to $h$ small enough, $k\geq 1$ such that $E_k^1\leq m_2-\frac12\alpha$, and $\lambda \le m_2 - \alpha$.
\vskip 0.2cm
By the implicit function theorem, it follows that the $k$-th eigenvalue $\lambda_k$ of $\widetilde P^\sharp$ satisfies,
$$
\lambda_k = E_{1,k} +\Oo (h^2),
$$
uniformly with respect to $h>0$ small enough and to $k=\Oo (h^{-1})$, such that $E_{1,k}\leq m_2 -\frac12\alpha$.

\subsection {Case 2: $\lambda \in [m_2-\alpha, m_2+\alpha ]$ with $\alpha >0$ small enough}

We denote by $\phi_1,\dots,\phi_n$ an orthonormal family of  eigenfunctions of $\widetilde P_1$ with eigenvalues in the interval $[m_2-2\alpha ,\, m_2+2\alpha]$ and by $\psi_1,\dots, \psi_m$ an orthonormal family of  eigenfunctions of $\widetilde P_2$ with eigenvalues in the interval $[m_2,\, m_2 +2\alpha ]$ (in particular, we have $n,m=\Oo (h^{-1})$).
\vskip 0.2cm
For $\alpha \oplus \beta\in \C^n\oplus \C^m$, we set,
$$
R_-(\alpha\oplus\beta ):= \alpha\cdot\phi \oplus \beta\cdot\psi \in \Hh^\sharp,
$$
where we have used the notation,
$$
\alpha\cdot\phi:= \sum_{k=1}^n \alpha_k\phi_k\quad ;\quad \beta\cdot\psi := \sum_{\ell=1}^m \beta_\ell\psi_\ell.
$$
We also denote by $R_+$ the adjoint of $R_-$, given by,
$$
R_+(u\oplus v) = (\langle u,\phi_k\rangle)_{1\leq k\leq n} \oplus (\langle v,\psi_\ell\rangle)_{1\leq \ell\leq m}.
$$

Then, we consider the operator valued matrix,
\be
G(\lambda) = \left(
\begin{array}{cc}
\widetilde P^\sharp -\lambda &R_- \\
R_+ &  0
\end{array} \right),
\ee
on
\be
\Hh^\sharp  \oplus \C^n \oplus\C^m,
\ee
with domain $(H^2\cap H^1_0)(\R_+)\oplus (H^2\cap H^1_0)(\R_+)\oplus \C^n \oplus\C^m$,
and we want to know whether $G(\lambda)$ is invertible.

We denote by $\Pi_1$ and $\Pi_2$ the orthogonal projections on the subspaces $S_n$ and $S_m$ of $L^2(\R_+)$ spanned by the eigenfunctions $\phi_1, \dots, \phi_n$ and $\psi_1, \dots, \psi_m$ respectively, and we set,
$$
\Pi:= \left(
\begin{array}{cc}
\Pi_1 &0 \\
0 &  \Pi_2
\end{array} \right)\, ;\, \Pi^\bot =  \left(
\begin{array}{cc}
\Pi_1^\bot &0 \\
0 &  \Pi_2^\bot
\end{array} \right):= \left(
\begin{array}{cc}
1-\Pi_1 &0 \\
0 &  1-\Pi_2
\end{array} \right).
$$

We first prove,
\begin{lemma}\sl For $\lambda\in[m_2-\alpha, m_2+\alpha]$, the operator $\Pi^\bot \widetilde P^\sharp \Pi^\bot-\lambda =:\widetilde P^\sharp_\bot -\lambda$ is invertible on the range $\mbox{Ran }\Pi^\bot$ of $\Pi^\bot$, and its inverse $(\widetilde P^\sharp_\bot -\lambda)^{-1}$ is uniformly bounded.
\end{lemma}
\begin{proof}
We have,
$$
\Pi^\bot (\widetilde P^\sharp -\lambda)\Pi^\bot=\left(
\begin{array}{cc}
\Pi_1^\bot (\widetilde P_1-\lambda)\Pi_1^\bot &h\Pi_1^\bot A_0 \Pi_2^\bot \\
h\Pi_2^\bot A_0^* \Pi_1^\bot &  \Pi_2^\bot (\widetilde P_2-\lambda)\Pi_2^\bot
\end{array} \right),
$$
and, denoting by $\widetilde P_j^\bot$ the restriction of $\widetilde P_j$ to $\mbox{Ran }\Pi_j^\bot$, we know that $\widetilde P_j^\bot -\lambda$ is invertible, and it is standard to show that its inverse is uniformly bounded from $\mbox{Ran }\Pi_j^\bot$ to $\mbox{Ran }\Pi_j^\bot \cap (H^2\cap H_1^0)(\R_+)$, if one takes the $h$-dependent norm on $H^2(\R_+)$ defined by: $\Vert u\Vert^2_{H^2} := \Vert h^2\Delta u\Vert_{L^2}^2 + \Vert u\Vert_{L^2}^2$. As a consequence, $A_0 \Pi_2^\bot (\widetilde P_2^\bot -\lambda )^{-1}\Pi_2^\bot$ and $A_0 \Pi_1^\bot (\widetilde P_1^\bot -\lambda )^{-1}\Pi_1^\bot$ are uniformly bounded on $L^2(\R_+)$ (together with their adjoint), and we find,
$$
\Pi^\bot (\widetilde P^\sharp -\lambda)\Pi^\bot \left(
\begin{array}{cc}
(\widetilde P_1^\bot -\lambda)^{-1} &0 \\
0 &  (\widetilde P_2^\bot -\lambda)^{-1}
\end{array} \right)\Pi^\bot = \Pi^\bot (1+\Oo (h))\Pi^\bot ;
$$
$$
\Pi^\bot \left(
\begin{array}{cc}
(\widetilde P_1^\bot -\lambda)^{-1} &0 \\
0 &  (\widetilde P_2^\bot -\lambda)^{-1}
\end{array} \right)\Pi^\bot (\widetilde P^\sharp -\lambda)\Pi^\bot= \Pi^\bot (1+\Oo (h))\Pi^\bot .
$$
Thus, the result follows by taking the restriction to $\mbox{Ran }\Pi^\bot$, and by using the Neumann series in order to inverse $\Pi^\bot (1+\Oo (h))\Pi^\bot\left|_{\mbox{Ran }\Pi^\bot}\right. =(1+\Pi^\bot \Oo (h))\left|_{\mbox{Ran }\Pi^\bot}\right.$.
\end{proof}

Using the previous lemma, it is easy to show that $G(\lambda)$ is invertible, and to check that its inverse is given by,
$$
G(\lambda)^{-1} =\left(
\begin{array}{cc}
\Pi^\bot(\widetilde P^\sharp_\bot -\lambda)^{-1}\Pi^\bot & (1-\Pi^\bot(\widetilde P^\sharp_\bot -\lambda)^{-1}\Pi^\bot \widetilde P^\sharp)R_- \\
R_+(1-\widetilde P^\sharp \Pi^\bot(\widetilde P^\sharp_\bot -\lambda)^{-1}\Pi^\bot) &
\lambda-Q(\lambda)
\end{array} \right)
$$
with,
\bee
\label{defQ}
Q(\lambda):=R_+\widetilde P^\sharp(1 -\Pi^\bot(\widetilde P^\sharp_\bot -\lambda)^{-1}\Pi^\bot \widetilde P^\sharp)R_-.
\eee
In particular, $Q(\lambda)$ is an $(n+m)\times (n+m)$ matrix with $n,m=\Oo(h^{-1})$.

\begin {proposition}\sl
 \label{Lem8} The matrix $Q(\lambda)$ satisfies,
\be
Q(\lambda)=
\mbox{\rm diag } (E_{1,1}, \dots, E_{1,n}, E_{2,1}, \dots, E_{2,m})
+ S(\lambda),
\ee
where $E_{1,j}, E_{2,k} \in [m_2-2\alpha ,\, m_2+2\alpha]$ are the eigenvalues associated with $\phi_j$ and $\psi_k$, respectively, and
with,
$$
\Vert S(\lambda)\Vert  + \Vert \frac{d}{d\lambda} S(\lambda) \Vert =\Oo (h^2),
$$
in the sense of the norm of operators on $\C^{n+m}$, and uniformly with respect to $h>0$ small enough and $n,m=\Oo(h^{-1})$.
\end {proposition}

\begin {proof}
Since $R_+\Pi^\bot=0$ and $\Pi^\bot R_-=0$, by (\ref{defQ}), we have,
\bee
\label{Q(lambda)bis}
Q(\lambda)=R_+\widetilde P^\sharp R_- - R_+\Pi \widetilde P^\sharp \Pi^\bot(\widetilde P^\sharp_\bot -\lambda)^{-1}\Pi^\bot  \widetilde P^\sharp\Pi R_-,
\eee
\bee
\label{dQ(lambda)}
\frac{d}{d\lambda}Q(\lambda)= R_+\Pi \widetilde P^\sharp \Pi^\bot(\widetilde P^\sharp_\bot -\lambda)^{-2}\Pi^\bot  \widetilde P^\sharp\Pi R_-,
\eee
and, since $\Pi_j \widetilde P_j \Pi_j^\bot=0$ ($j=1,2$),
\bee
\label{commutPiP}
\Pi \widetilde P^\sharp \Pi^\bot =\left(
\begin{array}{cc}
0& h\Pi_1 A_0\Pi_2^\bot \\
h\Pi_2 A_0^*\Pi_1^\bot &0
\end{array}
\right).
\eee
Moreover, using that $\Vert \widetilde P_j\Pi_j\Vert_{{\mathcal L}(L^2)} \leq | m_2| +2\alpha$ and the ellipticity of $\widetilde P_j$, it is easy to see that both $A_0^*\Pi_1$ and $A_0\Pi_2$ are uniformly bounded, thus so are their adjoints $\Pi_1A_0$ and $\Pi_2 A_0^*$, and we deduce from (\ref{Q(lambda)bis})-(\ref{commutPiP}) (plus the  fact that $\Vert R_\pm\Vert \leq 1$),
\bee
\label{estQlambda}
Q(\lambda)=R_+\widetilde P^\sharp R_- +\Oo(h^2)\quad ; \quad \frac{d}{d\lambda}Q(\lambda)=\Oo (h^2).
\eee
Therefore, in order to complete the proof of Proposition \ref{Lem8}, it is enough to show,

\begin{lemma}\sl
\label{estmicro}
For all $N\geq 0$, there exists  a constant $C_N>0$ such that, for all $j\in\{ 1, \dots, n\}$ and $k\in\{ 1, \dots, m\}$, one has,
$$
|\langle A_0\phi_j,\psi_k\rangle| + |\langle A_0\psi_k,\phi_j\rangle| \leq C_N h^N.
$$
\end{lemma}
\begin{proof}
We use the equations,
\bee
\label{equations}
(\widetilde P_1-E_{1,j})\phi_j=0 \quad ; \quad (\widetilde P_2-E_{2, k})\psi_k=0.
\eee
At first, we observe that, for $R$ close enough to 0 (say, $0<R<r_0$), and $R$ large enough (say, $R>R_0$), both $W_1(R)-E_{1,j}$ and $W_{2}(R) - E_{2,k}$ remain greater than some fix constant $C>0$. Therefore, by standard Agmon estimates (see, e.g., \cite{Ma1}, Chapter 3, exercise 8), it is easy to show that, for $h$ small enough,
\bee
\label{agmon}
\Vert \phi_j\Vert_{H^s((0,r_0)\cup (R_0, +\infty))}+\Vert \psi_k\Vert_{H^s((0,r_0)\cup (R_0, +\infty))}\leq e^{-c_0/h},
\eee
where the positive constant $c_0$ does not depend on $j,k =\Oo (h^{-1})$, and $s\geq 0$ is arbitrary.
\vskip 0.2cm
For $\ell=1,2$, we set,
$$
\Sigma_\ell:= \{ (R,R^*) \in \R_+ \times \R\, ;\, \widetilde p_\ell(R,R^*)\in [m_2-2\alpha, m_2+2\alpha]\}
$$
(where we have used the notation $\widetilde p_\ell(R,R^*):= (R^*)^2 +\widetilde W_\ell(R)$), and we chose
$\chi_\ell\in C_0^\infty ((\frac12 r_0, 2R_0)\times \R)$, supported near  $\Sigma_\ell$, such that $\chi_\ell =1$ in a neighborhood of $\Sigma_\ell$. We also fix $\chi_0=\chi_0(R)\in C_0^\infty (\frac12 r_0,2R_0)$,  such that $\chi_0=1$ near $[r_0, R_0]$.
\vskip 0.2cm

Then, using standard pseudodifferential calculus,  for any $E\in [m_2-2\alpha, m_2+2\alpha]$ one can construct a symbol $q_\ell(E)=q_\ell(E, R,R^* ; h) \in S(\langle R^*\rangle^{-2})$, supported in  $(\frac12 r_0, 2R_0)\times \R$
and depending smoothly on $E$,  such that,
\bee
\label{parametrix}
q_\ell(E) \# (\widetilde p_\ell- E) (R,R^*) \sim \chi_0(R)(1-\chi_\ell(R,R^*)).
\eee
Here, $\#$ stands for the Weyl-composition of symbols, and the asymptotic equivalence holds in $S(1)$, uniformly with respect to  $E\in [m_2-2\alpha, m_2+2\alpha]$ (see, e.g., \cite{Ma1}). Then,  first multiplying  (\ref{equations}) by $\chi_0$, then, commuting $\chi_0$ and $\widetilde P_j$, and finally applying  the usual  Weyl-quantization of $q_\ell(E)$ (with $E=E_{1,j}, E_{2,k}$, respectively),
we deduce from (\ref{equations}), (\ref{agmon}) and (\ref{parametrix}),
\begin{eqnarray}
&& \Vert (1-\chi_1(R, hD_R))\chi_0\phi_j\Vert_{H^s} =\Oo (h^\infty); \label{estphij}\\
&& \Vert (1-\chi_2(R, hD_R))\chi_0\psi_k\Vert_{H^s} =\Oo (h^\infty)\label{estpsik}
\end{eqnarray}
uniformly with respect to $j,k$ (here, $\chi_\ell(R, hD_R)$ stands for the Weyl-quantization of $\chi_\ell$).

\vskip 0.2cm
Now, if $\alpha$ is taken sufficiently small, the sets $\Sigma_1$ and $\Sigma_2$ are disjoints, and thus the supports of $\chi_1$ and $\chi_2$ can be taken disjoints, too. Since they are also disjoints from $\mbox{\rm Supp }(1-\chi_0)$, one can find $\chi_3 \in C^\infty ( \R_+\times\R ; [0,1])$, supported in $(\frac12 r_0, 2R_0)\times \R$, such that the family  $\{\chi_1, \chi_2, \chi_3 \}$ forms a partition of unity on $\mbox{\rm Supp }\chi_0 \times \R$. In particular, on $L^2(\R_+)$, one has,
$$
1-\chi_0 (R) +\sum_{\ell=1}^3 \chi_\ell(R, hD_R)\chi_0(R) =I,
$$
and now, it is clear that, inserting this microlocal partition of unity in the products $\langle A_0\phi_j,\psi_k\rangle$ and $\langle A_0\psi_k,\phi_j\rangle$, the estimates (\ref{agmon}), (\ref{estphij}) and (\ref{estpsik}) give the required result.
\end{proof}

\begin{remark}\sl
Actually, following more precisely the construction of $\widetilde H_e^\mu$ made in Section 4, one can prove that the functions $\widetilde W_j$ ($j=1,2$) and $a_0$ depend in an analytic way of $R$ in a neighborhood of the relevant classically allowed region $\{ \widetilde W_1(R) \leq m_2+2\alpha\}$. As a consequence, one can use the standard microlocal analytic techniques in this region (see, e.g., \cite{Sj, Ma1}), and obtain the existence of a constant $c_0>0$ (independent of $j,k$), such that,
$$
|\langle A_0\phi_j,\psi_k\rangle| + |\langle A_0\psi_k,\phi_j\rangle| \leq e^{-c_0/h}.
$$
\end{remark}
{\it Completion of the proof of the proposition}: Since the matrix,
$$
R_+\widetilde P^\sharp R_- - \mbox{\rm diag } (E_{1,1}, \dots, E_{2,m}) =
\left(
\begin{array}{cc}
0 & (\langle A_0\psi_k,\phi_j) \\
 (\langle A_0^*\phi_j,\psi_k\rangle)&  0
\end{array} \right)
$$
is of size $\Oo (h^{-1})$, Lemma \ref{estmicro} implies that it has a norm $\Oo (h^\infty)$ on $\C^{n+m}$, uniformly with respect to $n,m$. Thus, Proposition \ref{Lem8} is a consequence of (\ref{estQlambda}).
\end {proof}

By the Min-Max principle, it results from Proposition \ref{Lem8} that, for $\lambda\in [m_2-\alpha, m_2+\alpha]$, the eigenvalues $g_1(\lambda), \dots , g_{m+n}(\lambda)$ of $Q(\lambda)$ satisfy,
\begin{eqnarray}
&& \{ g_1(\lambda), \dots, g_{m+n}(\lambda)\} =\{ E_{1,1}, \dots, E_{1,n}, E_{2,1}, \dots, E_{2,m}\} +\Oo (h^2);\nonumber \\
&& \lambda\mapsto g_\ell (\lambda) \mbox{ is Lipschitz continuous } (\ell=1, \dots n+m); \nonumber\\
\label{derivgell}
&&\left| \frac{dg_\ell}{d\lambda} \right| =\Oo (h^2)\,\, a.e.\,\,  (k=1, \dots n+m).
\end{eqnarray}
Note that the values $E_{1,1}, \dots, E_{1,n}$ are at a distance of order $h$ from each other, and the same is true for the values $E_{2,1}, \dots, E_{2,m}$. So the only problem that may appear in the computation of $g_\ell(\lambda)$ is when, along some sequence $h=h_j\rightarrow 0_+$, two values $E_{1,j}$ and $E_{2,k}$ become closer than $\Oo(h^2)$. But, in that case, the two corresponding values of $g_\ell(\lambda)$ are given by,
\bee
\label{eigenmat2x2}
g_\ell(\lambda) = \frac12\left(E_{1,j}+E_{2,k} +h^2r_1 \pm \sqrt{(E_{1,j}-E_{2,k}+h^2r_2)^2+h^4r_3^2}\right),
\eee
with $r_t =r_t(\lambda)$ smooth, $r_t=\Oo (1)$, $dr_t/d\lambda =\Oo (1)$ ($t=1,2,3$). Thus, actually, a correct ($h$-depending) indexing of the $g_\ell$'s make them smooth functions of $\lambda$, and then (\ref{derivgell}) becomes true everywhere.

\vskip 0.2cm
Anyway, (\ref{derivgell}) is enough to insure that all the values of $\lambda\in [m_2-\alpha, m_2+\alpha]$ such that $\lambda\in \mbox{\rm Sp }Q(\lambda)$ verify,
$$
\mbox{\rm dist }(\lambda, \{ E_{1,1}, \dots, E_{1,n}, E_{2,1}, \dots, E_{2,m}\} )=\Oo (h^2),
$$
and, conversely, at any $E\in \{ E_{1,1}, \dots, E_{1,n}, E_{2,1}, \dots, E_{2,m}\} \cap [m_2-\alpha +Ch^2, m_2+\alpha -Ch^2]$ ($C>0$ large enough), can be associated a unique  $\lambda\in [m_2-\alpha , m_2+\alpha ]$ such that $\lambda\in \mbox{\rm Sp }Q(\lambda)$.
\vskip 0.2cm
Finally, using the fact that, by construction, the eigenvalues of $\widetilde P^\sharp$ that lie in  $[m_2-\alpha , m_2+\alpha ]$ coincide with the solutions there of $\lambda\in \mbox{\rm Sp }Q(\lambda)$, and summing up with the results of Subsection \ref{sub511} and Proposition \ref{compPmuPD}, we finally obtain,

\begin{theorem}
For $h >0$ small enough  the resonances of $P^\sharp$ with real part in $[m_1,m_2+\alpha]$ and with imaginary part $<< |h\ln h|$, coincide, up to $\Oo (h^2)$ error-terms, with  eigenvalues of the Dirichlet realizations of $P_{1}$ and $P_{2}$ on $(0,R_{1,M})$, where $R_{1,M}>0$ is the point where $W_1$ admits a local maximum with value greater than $m_2$.
\end{theorem}

\section {Comparison between the spectrum of the operators $P_\mu^\sharp$ and $\widetilde H_\mu^0$}
\label{secV}
Here we prove,
\begin{proposition}\sl Let $\alpha >0$ fixed small enough, and
let ${\mathcal J}\subset (0,1]$, with $0\in\overline{\mathcal J}$, such that there exists $\delta >0$ such that,
\begin{eqnarray}
{\rm Sp} (P^\sharp_D)\cap [m_2+\alpha-2\delta h, m_2+\alpha+2\delta h] =\emptyset.
\end{eqnarray}
Set,
$$
\Omega (h) := \{ z\in\C\,; \, {\rm dist}(\re z, [m_1,m_2+\alpha])<\delta h, \, |\im z| < C^{-1}h\ln\frac1{h}\},
$$
with $C>0$ a large enough constant. Then, there exists a bijection,
$$
b\, :\, {\rm Sp} (P_\mu^\sharp)\cap \Omega (h) \rightarrow {\rm Sp} (\widetilde H_\mu^0)\cap \Omega (h),
$$
such that,
$$
b(\lambda )-\lambda ={\mathcal O}(h^2),
$$
uniformly for $h\in {\mathcal J}$.
\end{proposition}
\begin{remark}\sl
As before, by slightly moving the parameter $\alpha$, one can actually reach all the values of $h>0$ small enough.
\end{remark}
\begin{proof}
By Corollary \ref{CorThm3}, it is enough to prove that, for any $z\in\Omega (h)$, there exists a bijection
$$
b_z\, :\, {\rm Sp} (P_\mu^\sharp)\cap \Omega (h) \rightarrow {\rm Sp} (\widetilde P_\mu^0(z))\cap \Omega (h),
$$
such that,
$$
b_z(\lambda )-\lambda ={\mathcal O}(h^2),
$$
uniformly for $h\in {\mathcal J}$ and $z\in \Omega (h)$.

By (\ref{Pmu}) we have,
\bee
\label{PmuQmu}
\widetilde P_\mu^0(z) = Q^0_\mu + h^2{\mathcal B}_\mu^0,
\eee
where $B_\mu^0$ stands for the restriction of $ {\mathcal B}_\mu(\RR,h \D_\RR ; z,h )$ to $Ker ({\mathbf L_R})$, and where we have set,
$$
Q_\mu^0:=Q_\mu\left\vert_{Ker ({\mathbf L_R})}\right.,
$$
$$
Q_\mu:=-h ^2 \Ss_\mu \Delta_\RR \Ss_\mu^{-1} +{\mathcal M}_\mu(R)+h {\mathcal A}_\mu (\RR,h \D_\RR).
$$
By passing in polar coordinates, and by conjugating $Q_\mu$ with the transform
$$
L^2( \R_+; R^2 dR)\otimes L^2(S^2)\ni \psi \mapsto R\psi \in L^2(\R_+; dR)\otimes L^2(S^2),
$$
 we see that $Q_\mu^0$ is unitarily equivalent to,
 $$
 \widetilde Q_\mu^0:= h ^2 \Ss_\mu \D_R^2 \Ss_\mu^{-1} +{\mathcal M}_\mu (R)+h {\mathcal A}_\mu^0 (R,h \D_R)
 $$
 on $L^2(\R_+; dR)$ with Dirichlet boundary condition at $R=0$ (the notations are those of the previous section, in particular (\ref{abusenot})).

We first have,
\begin{lemma}\sl
There exist $\delta_0>0$ and  a bijection,
$$
b_0\, :\, {\rm Sp} (P_\mu^\sharp)\cap \Omega (h) \rightarrow {\rm Sp} (Q_\mu^0)\cap \Omega (h),
$$
such that,
$$
b_0(\lambda )-\lambda ={\mathcal O}(e^{-\delta_0/h}),
$$
uniformly for $h\in {\mathcal J}$.
\end{lemma}
\begin{proof} This is just a slight modification of the proof of Proposition \ref{compPmuPD}. Indeed, using (\ref{StimaBasso}), we see that the proof can be repeated exactly in the same way by substituting $\widetilde Q_\mu^0$ to $P_\mu^\sharp$. Thus, both the spectra of $\widetilde Q_\mu^0$ and $P_\mu^\sharp$ are close to that of $P_D^\sharp$ up to  exponentially small error terms, and since $\widetilde Q_\mu^0$ and $ Q_\mu^0$ have the same spectrum, the result follows.
\end{proof}
Therefore, it only remains to compare the spectra of $\widetilde P_\mu^0(z)$ and $Q_\mu^0$.
\vskip 0.2cm
For any fixed integer $k\geq 1$, let us denote by $E_k=E_k(h)$ the $k$-th eigenvalue of $P_{2}$.
By the previous lemma and Remark \ref{remecarts}, we see that, if we fix $\delta >0$ sufficiently small, then,  the disc $\{ \lambda\in\C\,;\, |\lambda -E_k(h)|\leq \delta h\}$ contains at most two eigenvalues of $Q_\mu^0$ (for $h>0$ small enough) and, on the set ${\mathcal J}_k$ of those values of $h$ for which it contains two eigenvalues, the domain $\{ \lambda\in\C\,;\, \delta h< |\lambda -E_k(h)|\leq 2\delta h\}$ does not meet ${\rm Sp} (Q_\mu^0)$.

Then, for $k\geq 1$ we define,
\bee
\label{petitcontour}
\gamma_k (h) =\left\{\begin{array}{ll}
\{ \lambda\in\C\,;\, |\lambda -E_k(h)|=3 \delta h/2\} \,\mbox{ if } \,h\in {\mathcal J}_k;\\
\{ \lambda\in\C\,;\, |\lambda -E_k(h)|= \delta h/2\}
\,\mbox{ if }\, h\in {\mathcal J}\backslash {\mathcal J}_k.
\end{array}\right.
\eee
In the same way, the set $\{ \lambda\in\C\,;\, |\lambda -m_2|\leq \delta h\}$ contains at most one  eigenvalue of $Q_\mu^0$, and on the set ${\mathcal J}_0$ of those values of $h$ for which it contains one eigenvalue, the domain $\{ \lambda\in\C\,;\, \delta h< |\lambda -m_2|\leq 2\delta h\}$ does not meet ${\rm Sp} (Q_\mu^0)$. Then, we set,
$$
\gamma_0 (h) =\left\{\begin{array}{ll}
\{ \lambda\in\C\,;\,{\rm dist}(\lambda, [m_1,m_2])=3 \delta h/2\} \,\mbox{ if } \,h\in {\mathcal J}_0;\\
\{ \lambda\in\C\,;\,{\rm dist}(\lambda, [m_1,m_2])= \delta h/2\}
\,\mbox{ if }\, h\in {\mathcal J}\backslash {\mathcal J}_0.
\end{array}\right.
$$

When $\lambda\in\gamma_k (h)$ ($k\geq 0$), we see as in  the proof of Proposition \ref{compPmuPD} (see (\ref{resolvNeumann})) that the inverse of $\widetilde Q_\mu^0 -\lambda$ can be written as,
$$
(\widetilde Q_\mu^0 -\lambda)^{-1} = \chi_1 (\widetilde Q_D -\lambda)^{-1}\chi_2 + {\mathcal R}(\lambda),
$$
where $\widetilde Q_D$ is the Dirichlet realization of $\widetilde Q_\mu^0$ on $[0, R_{1,M}]$, $\chi_1, \chi_2$ are as in (\ref{resolvapp}), and ${\mathcal R}(\lambda)$ satisfies,
$$
\Vert {\mathcal R}(\lambda)\Vert_{{\mathcal L}({\mathcal H})} ={\mathcal O}(|h\ln h|^{-1})
$$
as in (\ref {pippobis}). \ Here, ${\mathcal H}$ is the  space introduced in the proof of Proposition \ref{compPmuPD}. \ In particular, we obtain,
$$
\Vert (\widetilde Q_\mu^0 -\lambda)^{-1}\Vert_{{\mathcal L}({\mathcal H})} ={\mathcal O}(h^{-1}),
$$
and thus, if we denote by ${\mathcal K}_0$ the space $Ker (L_\RR)$ endowed with the norm
$$
\Vert \psi\Vert_{{\mathcal K}_0}:= \Vert R\psi (R\omega)\Vert_{{\mathcal H}\otimes L^2(S^2)},
$$
we have,
\bee
\label{estresQmu0}
\Vert (Q_\mu^0 -\lambda)^{-1}\Vert_{{\mathcal L}({\mathcal K}_0)} ={\mathcal O}(h^{-1}).
\eee
On the other hand, thanks to (\ref{estbmu}), and by using the Calderon-Vaillancourt theorem (see, e.g., \cite{Ma1}), it is not difficult to show that the operator ${\mathcal B}_\mu^0$ is uniformly bounded on ${\mathcal K}_0$ (for instance, one can start by working on $C_0^\infty (\R^3\backslash 0)$ with the so-called right-quantization, in order to be able to pass in polar coordinates without problem, and then use a density argument and the fact that the polynomial weight used in the definition of ${\mathcal H}$ becomes trivial near $R=0$).

Therefore, we see on (\ref{PmuQmu}) that, for $h>0$ small enough and $\lambda\in\cup_{k\geq 0}\gamma_k (h)$, the operator $\widetilde P_\mu^0(z)-\lambda$ is invertible, and its inverse can be written as,
\bee
\label{estresPmu0}
(\widetilde P_\mu^0(z)-\lambda)^{-1}=(Q_\mu^0 -\lambda)^{-1}\left(I -h^2{\mathcal B}_\mu^0(Q_\mu^0 -\lambda)^{-1}+{\mathcal O}(h^2)\right),
\eee
in ${\mathcal L}({\mathcal K}_0)$.

For $k\geq 0$, we set,
\begin{eqnarray*}
&& \Pi_{P,k}(h) := \frac1{2i\pi }\oint_{\gamma_k (h)}(\lambda-\widetilde P_\mu^0(z))^{-1}d\lambda;\\
&& \Pi_{Q,k}(h) := \frac1{2i\pi }\oint_{\gamma_k (h)}(\lambda-Q_\mu^0)^{-1}d\lambda.
\end{eqnarray*}
In particular, for $k\geq 1$, the rank of $\Pi_{Q,k}(h)$ is 1 or 2, depending if $h\in {\mathcal J}_k$ or not. In both cases, (\ref{estresQmu0})-(\ref{estresPmu0}) show that the ranks of $\Pi_{P,k}(h)$ and $\Pi_{Q,k}(h)$ are identical, and that the eigenvalues of $\widetilde P_\mu^0(z)$ inside $\gamma_k(h)$ coincide to those of $Q_\mu^0$ up to ${\mathcal O}(h^2)$ (the computation is similar to that of  (\ref{eigenmat2x2})).

For $k=0$, the situation is even simpler, because we know that the eigenvalues of $Q_\mu^0$ that lie inside $\gamma_0(h)$ are simple and separated by a distance of order $h$, and the same result holds.

 Finally, for $\lambda\in \Omega(h)$ in the exterior of all the $\gamma_k(h)$'s, the estimate (\ref{estresQmu0}) is still valid, and thus so is (\ref{estresPmu0}).
Therefore, the spectral projector of $\widetilde P_\mu^0(z)$ on $\Omega(h)$ can be split into a finite sum of the $\Pi_{P,k}(h)$ (with a number of $k$'s that is ${\mathcal O}(h^{-1})$), and the previous arguments show that the eigenvalues of $\widetilde P_\mu^0(z)$ in $\Omega(h)$ coincide with those of $Q_\mu^0$ up to ${\mathcal O}(h^2)$.
\end{proof}

\section {Comparison between the spectrum of the operators $\widetilde H_\mu^0$ and $H_\mu^0$}
We have,
\begin{proposition}\sl
\label{compHtildeH}Let $\alpha >0$ fixed small enough, and
let ${\mathcal J}\subset (0,1]$, with $0\in\overline{\mathcal J}$, such that there exists $\delta >0$ such that,
\begin{eqnarray}
{\rm Sp} (P^\sharp_D)\cap [m_2+\alpha-2\delta h, m_2+\alpha+2\delta h] =\emptyset.
\end{eqnarray}
Set,
$$
\Omega (h) := \{ z\in\C\,; \, {\rm dist}(\re z, [m_1,m_2+\alpha])<\delta h, \, |\im z| < C^{-1}h\ln\frac1{h}\},
$$
with $C>0$ a large enough constant. Then, there exists $\delta_0>0$ and a bijection,
$$
b\, :\, {\rm Sp} (\widetilde H_\mu^0)\cap \Omega (h) \rightarrow {\rm Sp} (H_\mu^0)\cap \Omega (h),
$$
such that,
$$
b(\lambda )-\lambda ={\mathcal O}(e^{-\delta_0/h}),
$$
uniformly for $h\in {\mathcal J}$.
\end{proposition}
\begin{proof}
To prove this result, we use the arguments of Proposition 6.1 in \cite{MaMe}. \ In particular we have to check that condition (6.6) in \cite {MaMe}
holds true. \
We consider the oriented loop,
$$
\gamma(h):=\{ z\in\C\,;\,{\rm dist}(z, [m_1,m_2+\alpha])= \delta h/2\}.
$$
By (\ref{feschbachRes}) and the results of the previous section (in particular (\ref{estresQmu0})-(\ref{estresPmu0})), we already know that there exists some constant $C>0$ such that,
$$
\sup_{z\in \gamma (h)}\Vert (z-\widetilde H_\mu^0)^{-1}\Vert_{{\mathcal L}(Ker (\LL_\RR+\LL_\rr))} ={\mathcal O}(h^{-C}).
$$
We set,
$$
\widetilde \Pi_\mu:= \frac1{2i\pi}\oint_{\gamma (h)}(z-\widetilde H_\mu^0)^{-1} dz,
$$
and we denote by $F=F(R)\in C_0^\infty ((\frac4{M}, R_{1,M}); \R_+) $ a function that `fills the wells', in the same sense as in the proof of Proposition \ref{compPmuPD}, that is,
\bee
\label{tappo}
\inf_{(0,R_{1,M}]}(W_1+F)> m_2+\alpha.
\eee
Then, we set,
\bee
\widehat H^0_\mu = H_\mu^0 + F(R),
\eee
and we first prove,
\begin{lemma}\sl
\sl \label{Lem9}  Let $\Gamma(h)$ be the closure of the complex domain surrounded by $\gamma(h)$. Then, there exists a constant $C>0$ such that, for all $z\in \Gamma(h)$, one has,
$$
\left \|( \widehat H^0_\mu - z )^{-1} \right \|_{{\mathcal L}(Ker (\LL_\RR+\LL_\rr))} = \Oo (h ^{-C}),
$$
uniformly for $h>0$ small enough.
\end{lemma}
\begin {proof}
We use a standard method of localization that consists in decoupling the effects of the barrier from those of the remaining part of the operator (see, e.g., \cite{BCD}).

We fix $a,b>0$ such that $\frac3{M} < a < b < \frac4{M}$, and we denote by $J_I, J_E \in C^\infty (\R_+ ; [0,1])$ two functions satisfying,
\begin{eqnarray}
&& {\rm Supp} J_I\subset [0,b) \, ; \, {\rm Supp} J_E\subset (a, +\infty);\\
&& J_I = 1 \, \mbox{ on }\,  [0, a] \, ; \,  J_E =1\mbox{ on }\,  [b, +\infty);\\
&& J_I^2 +J_E^2 =1.
\end{eqnarray}
Next we denote by ${\mathcal H}_I$ the space $\{ u\left\vert_{\{ |\RR|\leq b\}} \, ; \, u\in Ker (\LL_\RR +\LL_\rr)\}\right.$ endowed with the standard $L^2$-norm, and ${\mathcal H}_E$ the space $Ker (\LL_\RR +\LL_\rr)$ endowed with the norm,
$$
\Vert u\Vert_{{\mathcal H}_E}:= \Vert Ru(R\omega, \rr)\Vert_{{\mathcal H}\otimes L^2(S_\omega^2)\otimes L^2(\R_\rr^3)}.
$$
We also define $\widetilde{\mathcal H}$ as the space $Ker (\LL_\RR +\LL_\rr)\}$ endowed with the norm,
$$
\Vert u\Vert_{\widetilde {\mathcal H}}:=\left( \Vert J_I u\Vert_{L^2}^2 + \Vert J_E u\Vert_{{\mathcal H}_E}^2\right)^{\frac12}.
$$
All these norms are clearly equivalent to the standard $L^2$-norms, with constants of equivalence of order $h^{\pm C}$ with $C>0$ constant, that is
\be
h^C \| u \|_{L^2} \le \| u \|_{{\mathcal H}_E} \le h^{-C} \| u\|_{L^2} \, .
\ee
In particular, it is enough to prove the result with $Ker (\LL_\RR +\LL_\rr)$ substituted by $\widetilde{\mathcal H}$.

Moreover, we have the so-called identifying operators,
$$
\begin{array}{cccc}
J\,: & {\mathcal H}_I\oplus {\mathcal H}_E & \rightarrow &\widetilde{\mathcal H}\\
{} & u\oplus v & \mapsto & J_I u+ J_E v\\
{}\\
\widetilde J\,: & \widetilde{\mathcal H}  & \rightarrow &{\mathcal H}_I\oplus {\mathcal H}_E\\
{} & w & \mapsto & J_I w\oplus J_E w
\end{array}
$$
that satisfy $J\widetilde J = {\bf 1}_{\widetilde{\mathcal H}}$, $\Vert \widetilde J\Vert =1$ (actually, $\widetilde J$ is an isometry, and is nothing but the adjoint of $J$ for the standard $L^2$-scalar product). By standard estimates on the transform $T$ (see, e.g., \cite{Ma1}), one can also easily see that $\Vert J\Vert ={\mathcal O}(1)$ uniformly as $h\rightarrow 0$.

Observing that the operator $\widehat H^0_\mu$ is differential with respect to $\RR$ (with operator-valued coefficients acting on $L^2(\R^3_\rr)$), we can consider the zero Dirichlet boundary condition at $|\RR |=b$ realizations $H_I$  of $\widehat H^0_\mu$ on ${\mathcal H}_I$ (note that it is nothing else but the restriction to $Ker (\LL_\RR +\LL_\rr)$ of the Dirichlet realizations of $H_\mu + F$ on $L^2(|\RR | <b)$). Finally, we set
$$
H_E := \widetilde H_\mu^0 + F,
$$
acting on ${\mathcal H}_E$, and
 we define,
$$
H_A:= H_I \oplus H_E,
$$
as an operator acting on ${\mathcal H}_I\oplus {\mathcal H}_E$. Then, setting,
$$
\Theta := \widehat H^0_\mu J - J H_A,
$$
it is elementary to check the identity,
\bee
\label{identRes}
(\widehat H^0_\mu - z)^{-1} = J (H_A-z)^{-1} \widetilde J - (\widehat H^0_\mu -z)^{-1}\Theta (H_A-z)^{-1} \widetilde J.
\eee
Using (\ref{feschbachRes}) and proceeding as in the proof of Proposition \ref{compPmuPD} and in Section \ref{secV}, we immediately obtain,
\bee
\label{estHE}
\Vert (H_E - z)^{-1}\Vert_{{\mathcal L}({\mathcal H}_E) }={\mathcal O}(|h\ln h|^{-1}).
\eee
On the other hand, since $b< 4/M$, by (\ref{Mmu})-(\ref{Mmu1}) we have,
$$
\re H_I \geq \frac{M}4 +\inf_R{\mathcal E}_1(R) \geq m_2+\alpha +1,
$$
if $M\geq 1$ has been chosen sufficiently large. As a consequence, for $z\in \Gamma (h)$, we have,
\bee
\label{estHI}
\Vert (H_I - z)^{-1}\Vert_{{\mathcal L}({\mathcal H}_I) }={\mathcal O}(1),
\eee
uniformly. From (\ref{estHE})-(\ref{estHI}), we deduce,
$$
\Vert (H_A - z)^{-1}\Vert_{{\mathcal L}({\mathcal H}_I\oplus {\mathcal H}_E) }={\mathcal O}(|h\ln h|^{-1}),
$$
and thus also, by standard estimates on the Laplacian,
\bee
\label{estHA}
\Vert \langle h\nabla_\RR\rangle (H_A - z)^{-1}\Vert_{{\mathcal L}({\mathcal H}_I\oplus {\mathcal H}_E) }={\mathcal O}(|h\ln h|^{-1}).
\eee
Now, we compute,
\be
\Theta (u\oplus v)&=& -h^2[\Delta_\RR, J_I]u - h^2[\Delta_\RR, J_E]v \\
&=& -h^2(2(\nabla_\RR J_I)\nabla_\RR + (\Delta_\RR J_I))u-h^2(2(\nabla_\RR J_E)\nabla_\RR + (\Delta_\RR J_E))v.
\ee
Therefore, we deduce from (\ref{estHA}) that we have,
$$
\Vert \Theta (H_A - z)^{-1}\widetilde J\Vert_{{\mathcal L}(\widetilde{\mathcal H})}={\mathcal O}(|\ln h|^{-1}).
$$
In particular, for $h$ small enough we obtain $\Vert \Theta (H_A - z)^{-1}\widetilde J\Vert_{{\mathcal L}(\widetilde{\mathcal H})}\leq 1/2$, and then, by using (\ref{identRes}) and, again, (\ref{estHA}), we finally deduce,
$$
\Vert (\widehat H^0_\mu - z)^{-1}\Vert_{{\mathcal L}(\widetilde{\mathcal H})}={\mathcal O}(|h\ln h|^{-1}),
$$
and the result follows.
 \end {proof}

Now, let (see Fig. \ref {Fig2})

\begin{eqnarray*}
&& R_1:= \min \{ R>0 \, ;\, W_1(R) =m_2+\alpha,\, W_1'(R) >0\};\\
&& R_2:=\min\{ R>R_1\, ;\, W_1(R) =m_2+\alpha, \, W_1'(R) <0\}.
\end{eqnarray*}
Let also $\chi_1(R), \chi_2 (R)\in C_0^\infty ([0, R_2])$, and $R_3\in (R_1,R_2)$, such that,
\be
\chi_1 =\chi_2 =1 \mbox{ near } [0, R_1]\quad ;\quad \mbox {supp} \chi_2\subset [0,R_3] \subset \subset \{ \chi_1=1\}.
\ee
We can also assume that the function $F(R)$ used in (\ref{tappo}) is such that $\chi_1 =\chi_2 =1$ in a neighborhood of Supp$F$, too.

\begin{figure}
\begin{center}
\includegraphics[height=12cm,width=12cm]{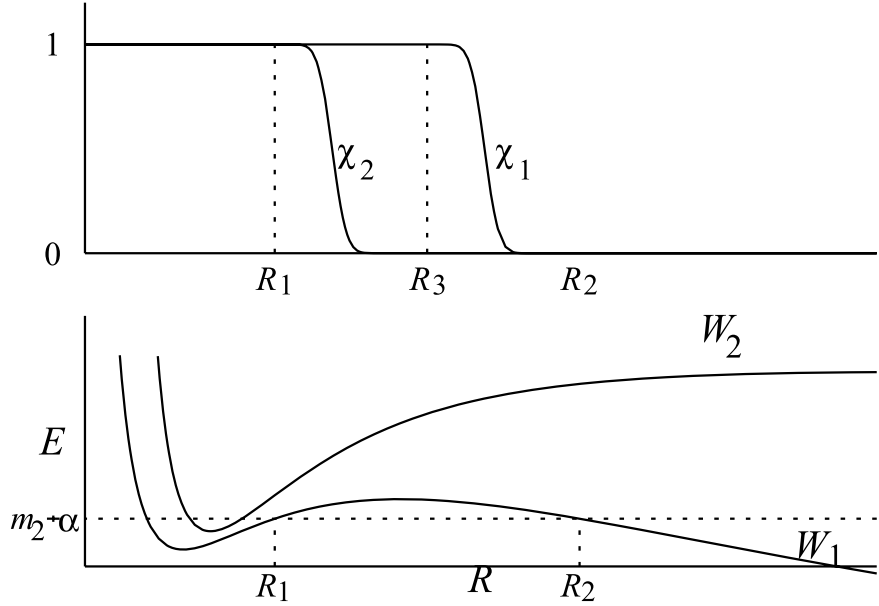}
\caption{Construction of the cut-off functions $\chi_1$ and $\chi_2$.}
\label {Fig2}
\end{center}
\end{figure}

Then, following \cite{HeSj2}, we set,
\be
\Rc_1 (z) = \chi_1 (H_D - z )^{-1} \chi_2 + (\widehat H_\mu^0 - z )^{-1} (1-\chi_2),
\ee
where $H_D $ is the Dirichlet realization of $H_\mu^0$ on $L^2(\{ |\RR |\le R_3 \} \times \R^3)$. \ Actually, $H_D$ does not depend on $\mu$ since ${\mathcal S}_\mu \equiv 1$ for $| \RR | < R_3$. \ Since $\chi_1 \chi_2 = \chi_2$, it follows from this definition (recalling that
$H_\mu^0 = \widehat H_\mu^0 - F(R)$), that we have,
\be
 (H_\mu^0 -z ) \Rc_1 (z) &=& [ h^2D_R^2 , \chi_1 ] (H_D - z)^{-1} \chi_2 +  \chi_2 + (H_\mu^0 -z ) (\hat H_\mu^0 - z )^{-1} (1-\chi_2)\\
&&  = [ h^2D_R^2 , \chi_1 ] (H_D - z)^{-1} \chi_2 +1 - F(R)
(\widehat H_\mu^0 - z )^{-1} (1-\chi_2 ) \\
&&  =1 + K_1(z) + K_2 (z),
\ee
where we have set,
\be
&& K_1 (z) := [ h^2D_R^2 , \chi_1 ] (H_D - z)^{-1} \chi_2 \\
&& K_2 (z) :=- F(R)(\widehat H_\mu^0 - z )^{-1} (1-\chi_2 ).
\ee
Here, we observe that,
$$
{\rm Supp}\hskip 1pt (D_R\chi_1)\cap {\rm Supp}\hskip 1pt (\chi_2)=\emptyset = {\rm Supp}\hskip 1pt F \cap {\rm Supp}\hskip 1pt (1-\chi_2).
$$
Moreover, introducing the notations,
\be
&& H_D = -h^2\Delta_\RR + P(R);\\
&& \widehat H_\mu^0 = -h^2 \Ss_\mu \Delta_{\RR} \Ss_\mu^{-1}+Q(R)
\ee
(where the two $R$-dependent operators $P(R)$ and $Q(R)$ act on the electronic variables only), we also have,
$$
{\rm Supp}\hskip 1pt (D_R\chi_1)\subset \{ P(R) > m_2+\alpha\} \quad ;\quad {\rm Supp}\hskip 1pt F \subset \{ \re Q(R) > m_2+\alpha\}.
$$
Therefore, proceeding as in \cite{HeSj2}, Section 9 (see also \cite{HeSj1}), by performing Agmon estimates, we deduce the existence of some constant $\delta >0$, such that,
\bee
\label{estK1K2}
\Vert K_1(z)\Vert + \Vert K_2(z)\Vert  = \Oo (e^{-\delta/h } ),
\eee
uniformly for $z\in \gamma (h)$ and $h>0$ small enough.

In the same way, setting,
\be
\Rc_2 (z) = \chi_2 (H_D - z )^{-1} \chi_1 +  (1-\chi_2)(\hat H_\mu^0 - z )^{-1},
\ee
we also have,
$$
\Rc_2 (z)(H_\mu^0 -z ) = 1+{\mathcal O}(e^{-\delta' /h}),
$$
with $\delta' >0$ constant. As a consequence, we deduce that $H_\mu^0 -z$ is invertible, and its inverse is given by,
\bee
\label{risolHmu0}
(H_\mu^0 - z)^{-1} = \Rc_1 (z) (1 +K_1 +K_2)^{-1}.
\eee
In particular, using Lemma \ref{Lem9} and the fact that $H_D$ is selfadjoint, and defining $\gamma_k(h)$ as in (\ref{petitcontour}), we conclude that, for all $z\in \gamma_k(h)$, we have,
\bee
\Vert (H_\mu^0 - z)^{-1}\Vert ={\mathcal O}(h^{-C}) \label {Equa75}
\eee
where $C>0$ is a constant.
\vskip 0.2cm
Now, we set,
\be
&& \Pi_{\mu,k} :=\frac1{2i\pi}\oint_{\gamma_k(h)} (H_\mu^0 - z)^{-1}  d z;\\
&& A_k : = \frac1{2i\pi}\chi_1 \oint_{\gamma_k(h)} (H_D - z)^{-1} \chi_2 d z.
\ee
(Note that, by construction, $\Vert A_k\Vert ={\mathcal O}(h^{-1})$ and $A_k$ is of rank at most 2.)

From (\ref{estK1K2})-(\ref{risolHmu0}) and Lemma \ref{Lem9} (plus the fact that $(\widehat H_\mu^0 - z)^{-1}$ is holomorphic inside $\gamma(h)$), we obtain,
\bee
\label{Pimu-A}
\Pi_{\mu,k} =  A_k + \Oo (e^{-\delta/h }).
\eee
In particular, since $\Pi_{\mu ,k}^2 =\Pi_{\mu ,k}$ and $\Vert \Pi_{\mu ,k}\Vert ={\mathcal O}(h^{-C})$, we deduce,
\bee
\label{A2=A}
A_k^2 =A_k+{\mathcal O}(e^{-\delta /2h}),
\eee
and thus, for any $\zeta \in\C$,
\bee
\label{eqopA}
(A_k-\zeta )(A_k+\zeta -1) = - \zeta (\zeta -1) + R_k , \, \Vert R_k\Vert= \Oo (e^{-\delta /2h }).
\eee
As a consequence, if  $\zeta \not= 0,1$ is fixed, then $(A_k-\zeta )$ is invertible, and we can consider the projection,
\be
\Pi_{A_k} := \frac1{2i\pi}\oint_{|\zeta -1|=1/2} (\zeta -A_k)^{-1} d \zeta .
\ee
Then, we prove,
\begin{lemma}\sl
\label{PiA-A}
One has,
$$
\Vert A_k-\Pi_{A_k}\Vert ={\mathcal O}(e^{-\delta/4h}).
$$
uniformly for $h>0$ small enough.
\end{lemma}
\begin{proof}
We write,
\be
\Pi_{A_k} -A_k = \frac {1}{2\pi i} \oint_{|\zeta -1|=1/2} \left [ (\zeta -A_k)^{-1} - (\zeta -1)^{-1}A_k \right ] d\zeta
\ee
and,
\be
(\zeta -A_k)^{-1} - (\zeta -1)^{-1} A_k = (\zeta -A_k)^{-1} \left(1-A_k+(\zeta -1)^{-1}(A_k-A_k^2)\right).
\ee
Moreover, by (\ref{eqopA}), we also have,
\bee
\label{invAk}
(A_k-\zeta )^{-1} = (A_k+\zeta -1) \left [ \zeta (1-\zeta )+ R_k \right ]^{-1}
\eee
In particular, $\Vert (A_k-\zeta )^{-1}\Vert ={\mathcal O}(h^{-1})$ uniformly on $\{|\zeta -1|=1/2\}$, and thus, using (\ref{A2=A}), we obtain,
\bee
\Pi_{A_k} - A_k = \Pi_{A_k} (1 -A_k) + \Oo (e^{-\delta/3h } ).
\eee
On the other hand, using (\ref{invAk}) (and the fact that $R_k=A_k^2-A_k$ commutes with $A_k$), we have,
\be
\Pi_{A_k} (1 -A_k)&=& \frac1{2i\pi}\oint_{|\zeta -1|=1/2}(A_k-\zeta )^{-1}(\zeta -1)d\zeta \\
&=& \frac1{2i\pi}\oint_{|\zeta -1|=1/2}(\zeta -1)(A_k+\zeta -1) \left [ \zeta (1-\zeta )+ R_k \right ]^{-1}d\zeta \\
&=&\frac1{2i\pi}\oint_{|\zeta -1|=1/2} \left( \zeta -1-\frac1{\zeta }R_k\right) \left ( \zeta (1-\zeta )+ R_k \right )^{-1}(A_k+\zeta -1)dz\\
&& + \frac1{2i\pi}\oint_{|\zeta -1|=1/2}\zeta^{-1}R_k \left ( \zeta (1-\zeta )+ R_k \right )^{-1}(A_k+\zeta -1)d\zeta \\
&=& \frac1{2i\pi}\oint_{|\zeta -1|=1/2}\zeta^{-1}R_k \left ( \zeta (1-\zeta )+ R_k \right )^{-1}(A_k+\zeta -1)d\zeta \\
&=& {\mathcal O}(e^{-\delta/4h }),
\ee
where we have used the fact that,
\be
\oint_{|\zeta -1|=1/2} \left( \zeta -1-\frac1{\zeta }R_k\right) \left ( \zeta (1-\zeta )+ R_k \right )^{-1}(A_k+\zeta -1)d\zeta \\
=\oint_{|\zeta -1|=1/2} -\zeta^{-1}(A_k+\zeta -1)d\zeta =0,
\ee
because the function inside the integral is analytic in the disc $\{|\zeta -1|\leq1/2\}$.
\end{proof}

We deduce from Lemma \ref{PiA-A} and (\ref{Pimu-A}) that we have,
\bee
\label{Pimu-PiA}
\Pi_{\mu,k} =  \Pi_{A_k} + \Oo (e^{-\delta/4h }).
\eee
Moreover, still by Lemma \ref{PiA-A}, we see that the restriction,
$$
A_k\left\vert_{{\rm Im} (\Pi_{A_k})}\right.\, :\, {\rm Im} (\Pi_{A_k})\rightarrow {\rm Im} (\Pi_{A_k})
$$
is invertible, and thus, since the rank of $A_k$ is at most 2, we deduce,
$$
{\rm Rank}(\Pi_{A_k})\leq 2.
$$
As a consequence, by (\ref{Pimu-PiA}), we obtain,
$$
{\rm Rank}(\Pi_{\mu,k})\leq 2.
$$
Then, we are exactly in the situation of \cite{MaMe} Proposition 6.1, (i)-(ii), and, setting,
$$
\widetilde \Pi_{\mu,k}:= \frac1{2i\pi}\oint_{\gamma_k (h)}(z-\widetilde H_\mu^0)^{-1} dz,
$$
we conclude that we have,
\bee
\Vert \Pi_{\mu,k} - \widetilde \Pi_{\mu,k}\Vert ={\mathcal O}(e^{-\delta'' /h}),
\eee
with $\delta'' >0$ constant. As before, the result on the interior of $\gamma_k(h)$ ($k\geq 1$) follows in a standard way, and one obtains the required comparison of the spectra of $H_\mu^0$ and $\widetilde H_\mu^0$ on
\be
{\mathcal B} = \{ {\rm Re}\, z \in [m_2, m_2+\alpha], \, |\im z| < C^{-1}|h\ln h|\}.
\ee
As before, the same arguments can be performed on $\{ Re z \in [m_1, m_2], \, |\im z| < C^{-1}|h\ln h|\}$ (in a simpler way, since the eigenvalues of $H_D$ are separated by a distance of order $\sim h$), and Proposition \ref{compHtildeH} follows.
\end{proof}

\begin {remark}
The same argument, using (\ref {Equa75}), also show that the spectrum of $H_{\mu}^0$ and $H_D$  on ${\mathcal B}$ coincides up to an exponentially small term. \ In particular, the eigenvalues of $H_D$ are real and then the resonances of $H_\mu^0$ have exponentially small part.
\end {remark}

\section {Main Result}

Here, by collecting the Propositions 5.1 and 6.1, and Theorem 4.8 it turns out our main result.

\begin {theorem} \label {TeoremaPrincipale}
Let $\alpha >0$ fixed small enough, and
let ${\mathcal J}\subset (0,1]$, with $0\in\overline{\mathcal J}$, such that there exists $\delta >0$ such that,
\be
{\rm Sp} (P^\sharp_D)\cap [m_2+\alpha-2\delta h, m_2+\alpha+2\delta h] =\emptyset.
\ee
Set,
$$
\Omega (h) := \{ z\in\C\,; \, {\rm dist}(\re z, [m_1,m_2+\alpha])<\delta h, \, |\im z| < C^{-1}h\ln\frac1{h}\},
$$
with $C>0$ a large enough constant. \ For $h >0$ small enough then the resonances of $H_\mu^0$ in $\Omega (h)$ coincide up to $\Oo (h^2)$ error-terms, with  eigenvalues of the Dirichlet realizations of $P_{1}$ and $P_{2}$ on $(0,R_{1,M})$, where $R_{1,M}>0$ is the point where $W_1$ admits a local maximum with value greater than $m_2$.
\end{theorem}

\begin {remark}
As done in \S 5, by slightly moving the parameter $\alpha$ one can actually reach all the values of $h >0$ small enough.
\end {remark}

\begin {remark}
We would point out that this result still holds true even in absence of the external field; in such a case we don't have resonances for $H_\mu^0$, but real eigenvalues.
\end {remark}

\end {document}